\documentclass[12pt,doublespace]{article}
\usepackage{setspace}
\usepackage{amsfonts}
\usepackage{amsmath}
\usepackage{graphicx}
\usepackage{algorithmic}
\usepackage{algorithm}
\usepackage{MnSymbol}
\usepackage{epstopdf}  
\usepackage{amsthm}   
\usepackage{xcolor,colortbl}
\usepackage{color}

\usepackage{tikz}
\usepackage{verbatim}
\newtheorem*{remark}{Remark}
\usetikzlibrary{fit,positioning}
\title{Sparse Kalman Filtering Approaches to Covariance Estimation from High Frequency Data in the Presence of Jumps}

\author{Michael Ho\footnotemark[2], Jack Xin\footnotemark[4]}
\begin{document}
\maketitle
\newenvironment{Prog}{\refstepcounter{Prog}\align}{\tag{P\theProg}\endalign}
\newenvironment{Prog*}{\align}{\endalign}
\renewcommand{\thefootnote}{\fnsymbol{footnote}}
\footnotetext[2]{Department of Mathematics, UC Irvine, Irvine, CA 92697, USA. Email: mtho1@uci.edu.}
\footnotetext[4]{Department of Mathematics, UC Irvine, Irvine, CA 92697, USA. Email: jxin@math.uci.edu.}

\renewcommand{\thefootnote}{\arabic{footnote}}
\newcommand{\overbar}[1]{\mkern 1.5mu\overline{\mkern-1.5mu#1\mkern-1.5mu}\mkern 1.5mu}
\newcommand{\Ell}[1]
{  \ell_{#1}  }
\newcommand{\supp}[1]
{  \textnormal{supp}({#1})  }
\newcommand{\suppEpp}[2]
{  \textnormal{supp}_{#2}({#1})  }
\newcommand{\sign}[1]
{  \textnormal{sgn}({#1})  }
\newcommand{\Elln}[2]
{  ||#1||_{\ell_{#2}}  }
\newcommand{\EllnW}[3]
{  ||#1||_{#2,\ell_{#3}}  }

\newcommand{\EllnBig}[2]
{  \Big|\Big|#1 \Big|\Big|_{\ ell_{#2}}  }
\newcommand{\E}
{ \mathbb{E} }
\newcommand{\seq}[2]
{\{#1_{n}\}_{n=1}^{#2} }
\newcommand{\Ft}
{\mathcal{F}_{t}}
\newcommand{\Yt}
{Y_{t}}
\newcommand{\hYt}
{\hat{Y}_{t}}
\newcommand{\SigEp}
{\Sigma^{\epsilon}}
\newcommand{\wTil}
{\tilde{w}}
\newcommand{\Rtil}
{\tilde{R}}
\newcommand{\vect}[1]
{\vec{#1}}
\newcommand{\partialDer}[1]
{ \frac{\partial}{\partial #1} }
\newcommand{\cov}
{ cov }
\newcommand{\lamSet}
{ \{\lambda_{i}(t)\}_{1\le i\le N, 2 \le t \le T  }}
\newcommand{\lamSetInv}
{ \{\lambda_{i}(t)^{-1}\}_{1\le i\le N, 2 \le t \le T  }}
\newcommand{\sigSet}
{ \{\sigma_{o,i}^{2}\}_{1\le i\le N  }}
\newcommand{\trace}
{ \textnormal{tr} }
\newcommand{\ci}{\mathrel{\text{\scalebox{1.07}{$\perp\mkern-10mu\perp$}}}}

\newtheorem{theo}{Theorem}
\newtheorem{theorem}[theo]{Theorem}
\newtheorem{lemma}[theo]{Lemma}
\newtheorem{prop}[theo]{Proposition}
\newtheorem{corr}[theo]{Corollary}

\newtheorem{Def}{Definition}
\newtheorem{Remark}{Remark}
\newtheorem{definition}{Definition}
\newtheorem{Conj}{Conjecture}
\newtheorem{example}{Example}
\setlength{\arrayrulewidth}{1pt}
\begin{abstract}

Estimation of the covariance matrix of asset returns from high frequency data is complicated by asynchronous returns, market microstructure noise and jumps.
One technique for addressing both asynchronous returns and market microstructure is the Kalman-EM (KEM) algorithm.   However the KEM approach assumes log-normal prices and
does not address jumps in the return process which can corrupt estimation of the covariance matrix.

In this paper we extend the KEM algorithm to price models that include jumps.   We propose two sparse Kalman filtering approaches to this problem.    In the first approach we develop a Kalman Expectation Conditional Maximization (KECM) algorithm
to determine the unknown covariance as well as detecting the jumps.   For this algorithm we consider Laplace and the spike and slab jump models, both of which promote sparse estimates of the jumps.
In the second method we take a Bayesian approach and use Gibbs sampling to sample from the posterior distribution of the covariance matrix under the spike and slab jump model.  Numerical results using simulated data show that each of these approaches
provide for improved covariance estimation relative to the KEM method in a variety of settings where jumps occur.
\end{abstract}
\section{Introduction}
 The covariance matrix of asset returns is an integral element of many financial optimization problems such as portfolio design.   For example, in minimum variance portfolio optimization the criterion for selecting the portfolio weights ,$w$, can be written as
\begin{align*}
\min_{w} w^{T}\Gamma w \\
\textnormal{s.t.} \sum_{i} w_{i} =1
\end{align*}
where $\Gamma$ is the covariance matrix of the asset returns.   Since the covariance matrix is usually unknown, the above criterion cannot be implemented exactly.   Instead an estimate of the covariance matrix, $\hat{\Gamma}$, is obtained and substituted into the portfolio optimization criterion.

A simple and intuitive approach to estimating the covariance matrix is to form a sample average of the covariance matrix using from past return data.  However  when a finite number of samples are used, covariance estimation errors will be present.   These errors can result in portfolio performance that departs significantly from the optimal performance under known statistics \cite{DeMiguelOneOverN,Barry1974,Jobson1980}.  Thus for portfolio optimization to be effective an accurate estimate of the covariance matrix is paramount.

Appealing to the law of large numbers covariance estimation errors can be reduced by using more data in the sample average estimate.  One approach to obtain more data is to simply increase the time window size when forming the sample covariance (e.g. use 1 year of data vs 3 months of data).
In order for this approach to be effective the additional data used in covariance estimation should be nearly identically distributed to future data.   If the data statistics are non-stationary then increasing the window's size to obtain more data may not improve portfolio performance as the additional data used in the covariance estimation may not be relevant to future returns.

 Another approach to obtaining more data is to sample at a higher frequency \cite{SahaliaHowOften} (e.g. 1 second update rate vs 1 day update rate) and maintain the sampling window size.   This approach is less vulnerable to non-stationary statistics but presents additional challenges unique to high frequency data.   For example,  high frequency data is subject to market microstructure noise \cite{LoAndrewBook} such as bid-ask bounce which can corrupt volatility and covariance estimates.   At higher frequencies the variance of the market microstructure noise can mask the true volatility of the asset returns if it is not accounted for \cite{BandiRussell,SahaliaHowOften}.     Asynchronous trading of assets observed at higher frequencies \cite{LoNonSync} further complicates covariance estimation as the standard sample average estimate assumes return data is available at each time instance.

 Many approaches have been proposed for estimating covariance matrices from high frequency data in the presence of asynchronous trading and microstructure noise.  For example, the refresh-time approach proposed in \cite{RefreshTime} addresses asynchronous trading by attempting to synchronize the return data by waiting for all assets to trade at least one time prior to forming a asset price vector used in covariance estimation.   One disadvantage of this approach is that much of the data is ignored while waiting for all assets to trade.   The pairwise refresh approach \cite{VastVolMatrix} uses more data by refreshing the covariance matrix element by element.   Here we form a $2 \times 1$ asset price vector every time period where two assets trade.   This allows for more data to be used but the resulting sample covariance matrix is not guaranteed to be positive semi-definite without applying additional corrections such as a projection method \cite{VastVolMatrix}.
 Another approach is the previous tick method employed in \cite{ZhangEpps} where a fixed sampling grid is defined and trade prices are approximated on that grid as the nearest previous trade price.

 To address both micro-structure noise and asynchronous returns, quasi-maximum likelihood estimators were proposed in \cite{SahaliaHighFreqCovariance,LiuQuasiMLE} that utilize pairwise refresh.  A two scale realized covariance (TSCV) approach was developed in \cite{ZhangEpps} where covariance estimates are obtained using both low frequency and high frequency sampling.   An approach based on Kalman filtering and the EM algorithm \cite{KalmanEMStocks}, models the true unobserved log-price process and observed prices as a discrete linear normal dynamical system.  Here the unobserved synchronous true price is treated as latent data and the EM algorithm is used to determine a maximum-likelihood estimate of the covariance.    A Bayesian version of the Kalman-EM approach where the posterior distribution of the covariance is approximated via an augmented Gibbs sampler is proposed in \cite{BayesianEM}.   This technique generates an estimate of the posterior distribution of the covariance which can then be used to obtain to a point estimate.

Each of the above techniques addressing micro-structure noise and asynchronous returns utilize a log-normal price model.   However, empirical return data often exhibits heavy tails that are better explained by a jump diffusion or stochastic volatility models.  Under these conditions the approaches which assume log-normal returns will yield sub-optimal results.   Techniques for addressing jumps have been proposed in the literature.  In \cite{FanJumps} the authors propose wavelet techniques for detecting jumps with an application to volatility estimation.  The jumps estimated using this approach are then removed from the observed data prior to volatility estimation.   In \cite{BoudtJump} a jump detector is employed to selectively remove data that contain jumps from the covariance estimation samples prior to TSCV.   Another technique proposed in \cite{BoudtJump2} is also robust to jumps but does not address market microstructure noise.

In this paper we extend the Kalman-EM approach in \cite{KalmanEMStocks} to discretized jump diffusion models by introducing two Kalman-ECM (KECM) approaches.  In our first KECM approach we model the jumps as Laplace distributed random variables.   Although the Laplace prior may seem to be an unnatural model for a jump process, we will see that the prior promotes a sparse posterior mode for the jumps by inducing an $\ell_{1}$ norm penalty on the jumps into the complete log-likelihood function.  Conditioned on other variables determining the posterior mode for the jumps is a convex $\ell_{1}$ norm penalized quadratic program which can be solved with a variety of fast techniques \cite{GoldStein1,ADMM,FISTA}.   In our second KECM approach we consider a more natural, but less tractable, spike and slab model for the jump process.

We also extend the Bayesian approach in \cite{BayesianEM} to jump models where jumps are modeled using a spike and slab prior \cite{SpikeSlab,BayesL1}.   Here we use Gibbs sampling to approximate the posterior of the jumps along with the unknown covariance matrix.   An estimate of the posterior mean of the covariance matrix is then obtained using the samples obtained from the posterior distribution.

The remainder of this paper is organized as follows.    In section 2 we introduce the models which form the basis for our covariance estimation approaches.   In section 3 and 4 we describe numerical algorithms for computing the covariance estimate with both the Laplace and spike and slab prior.   A performance evaluation of our proposed approach is presented in section 5 using simulated high frequency data.   A summary and conclusion are presented in section 6.


\section{High Frequency Return Modeling}\label{Sec:ModelHigh}
Suppose that we have $N$ assets where the true (or efficient) log price of the $n^{th}$ asset at time $t$ is $X_{n}(t)$.   Let $X(t)$ denote the $N \times 1$ vector of log prices for each asset at time $t$ and let $T$ denote the total number of time samples.   Here $X_{n}(t)$ can be viewed as the fundamental value of the asset in an efficient market without friction \cite{Roll84}.

We model the dynamics of the log prices using a discrete time jump diffusion model with a drift $D$
\begin{equation}\label{eqX}
  X_{i}(t)=X_{i}(t-1) + V_{i}(t) +\tilde{J}_{i}(t)Z_{i}(t)+D.
\end{equation}
Here we assume the following:
\begin{itemize}
  \item $V(t)$ is multivariate normally distributed with mean $0$ and covariance $\Gamma$
  \item $\tilde{J}_{i}(t)$ is normally distributed with zero mean and variance $\sigma_{j,i}^{2}(t)$
  \item  $Z_{i}(t)$ is Bernoulli distributed, with $Pr(Z_{i}(t)=0)=\zeta$
  \item $ \tilde{J}_{m}(t) \ci \tilde{J}_{n}(s),Z_{m}(t) \ci Z_{n}(s) \;\; , m \ne n$ and all $t,s$
  \item $\tilde{J},Z,V$ are jointly independent.
\end{itemize}
To simplify notation we denote the jump component as
\begin{equation}\label{eq:DefJump}
  J(t)=\tilde{J}(t)Z(t).
\end{equation}

In many markets trading of distinct assets does not occur simultaneously.   When trades occur asynchronously, current pricing data for all assets will not be observed.   For prices that are observed, market microstructure noise needs to be addressed.  Here transaction costs due to order processing expenses, inventory costs and adverse selection costs \cite{LoAndrewBook} add noise to the true efficient price.  Thus the true efficient price is not directly observed.   

Both asynchronous returns and microstructure noise can be captured in the following observation model
\begin{equation}\label{eqPriceObserved}
  Y(t) = \tilde{I}(t)X(t)+W(t)
\end{equation}
where
\begin{itemize}
  \item $\tilde{I}(t)$ is a``partial'' identity matrix where the rows corresponding to missing asset prices at time $t$ are removed
  \item $W(t)$ is normal distributed market microstructure noise with zero mean and covariance $\Sigma_{o}(t)=\tilde{I}(t)\Sigma_{o}'\tilde{I}(t)^{T}$.
\end{itemize}
Here $\Sigma_{o}'$ is a diagonal matrix $\textnormal{diag}(\sigma_{o,1}^{2}, \dots,\sigma_{o,N}^{2})$.
For purposes of this paper we shall assume that $\tilde{I}(t)$ is known, and that $\{W(t),X(t)\}_{t=1}^{T}$ are jointly independent.   In section \ref{sec:ECM_numerical} we will test our algorithms on simulated data where the microstructure noise and price innovation are statistically dependent.

\subsection{Conditional Distributions of Observations and Log-Prices}

Now we examine the joint probability distribution of $X(1:T)$,$Y(1:T)$,$J(2:T)$.   Here the notation $X(m:n)$ refers to the set $\{X(m),X(m+1),\dots,X(n)\}$.   We consider the case of when the parameters $D,\Gamma,\sigma_{o,i}^{2},\zeta$ and $\sigma_{j,i}^{2}$ are random variables with known prior distributions.   Details on our assumed priors are given in section \ref{sec:Priors}.

To determine the probability distribution we first note that our model in equations \eqref{eqX} and \eqref{eqPriceObserved} can be represented by the Bayesian network depicted in Figure \ref{Fig:BayesNetDet}.  From the Bayesian network we see that the following conditional independence properties hold
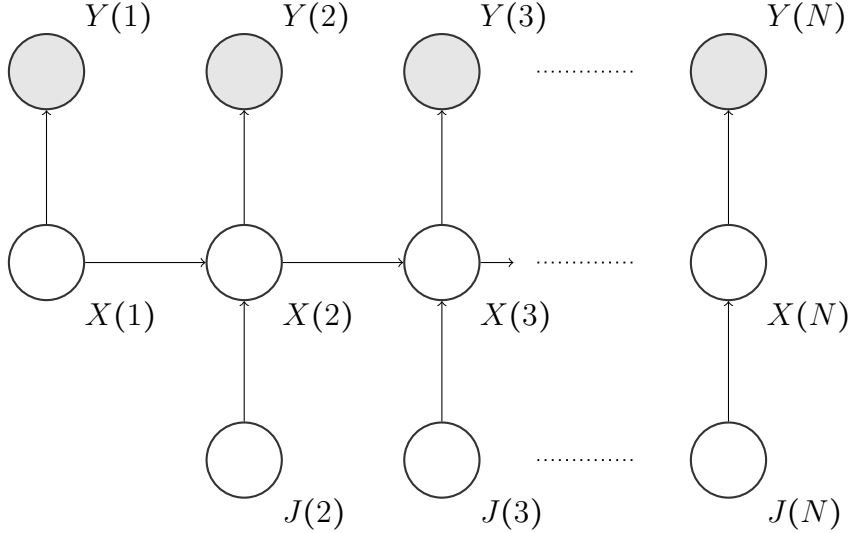
\begin{figure}\label{Fig:BayesNetDet}
\centering
\begin{tikzpicture}
\tikzstyle{main}=[circle, minimum size = 10mm, thick, draw =black!80, node distance = 16mm]
\tikzstyle{connect}=[-latex, thick]

  \node[main, fill = white!100] (x1) [label=below right:$X(1)$] { };
  \node[main] (x2) [right=of x1,label=below right:$X(2)$] { };
  \node[main] (x3) [right=of x2,label=below right:$X(3)$] { };
  \coordinate[right of =x3,node distance=0.5in](emptyCoord1);
  \coordinate[right of =emptyCoord1,node distance=0.5in](emptyCoord2);
  \node[main] (xN) [right of= emptyCoord2,label=below right:$X(N)$,node distance=0.5in] { };

  \node[main,fill = black!10] (y1)[above of =x1,label=above right:$Y(1)$,node distance=1in] { };
  \node[main,fill = black!10] (y2) [right=of y1,label=above right:$Y(2)$] { };
  \node[main,fill = black!10] (y3) [right=of y2,label=above right:$Y(3)$] { };
  \coordinate[right of =y3,node distance=0.5in](YemptyCoord1);
  \coordinate[right of =YemptyCoord1,node distance=0.5in](YemptyCoord2);
  \node[main,fill = black!10] (yN) [right of= YemptyCoord2,label=above right:$Y(N)$,node distance=0.5in] { };

  \node[main] (j2) [below=of x2,node distance=1in,label=below right:$J(2)$] { };
  \node[main] (j3) [right=of j2,label=below right:$J(3)$] { };

  \coordinate[right of =j3,node distance=0.5in](jemptyCoord1);
  \coordinate[right of =jemptyCoord1,node distance=0.5in](jemptyCoord2);
  \node[main] (jN) [right of= jemptyCoord2,label=below right:$J(N)$,node distance=0.5in] { };

  \coordinate[left of =YemptyCoord1,node distance=0.125in](YemptyCoord0);
  \coordinate[left of =emptyCoord1,node distance=0.125in](emptyCoord0);

  \draw[->,label] (x1) to (x2);
  \draw[dotted,line width=0.25mm] (emptyCoord1) to[dotted] (emptyCoord2);
  \draw[dotted,line width=0.25mm] (YemptyCoord1) to[dotted] (YemptyCoord2);
  \draw[dotted,line width=0.25mm] (jemptyCoord1) to[dotted] (jemptyCoord2);
  \draw[->] (x1) to (y1);
  \draw[->] (x2) to (y2);
  \draw[->] (x3) to (y3);
  \draw[->] (x2) to (x3);
  \draw[->] (xN) to (yN);
  \draw[->] (jN) to (xN);
  \draw[->] (j2) to (x2);
  \draw[->] (j3) to (x3);

  \draw[->] (x3) to (emptyCoord0);

\end{tikzpicture} 
\caption{Bayesian Network Representation of $(X,Y,J)$.   Observed variables are shaded.   Here the model parameters are not shown.}
\label{Fig:BayesNetDet}
\end{figure}
\begin{eqnarray*}
  Y(t) \ci J(s) |X(t) \;\; \forall s  \\
  Y(t) \ci X(s) |X(t) \;\; \forall s \ne t \\
  X(t) \ci X(s) |(X(t-1),J(t)) \;\;  \forall s < t-1 \\
  X(t) \ci J(s) |(X(t-1),J(t)) \;\;   \forall s \ne t.
\end{eqnarray*}

From the conditional independence implied by the Bayesian network we have that the probability distribution conditioned on the parameter values may be fully characterized as follows
\begin{eqnarray*}
  p(y(t)|x(1:T),\Sigma_{o}^{2}(t)) &\sim& \mathcal{N}(\tilde{I}(t)x(t),\Sigma_{o}(t)) \\
  p(x(t+1)|x(1:t),j(2:t+1),d,\Gamma) &\sim& \mathcal{N}(x(t)+j(t+1)+d,\Gamma) \\
  p(x(1)) &\sim& \mathcal{N}(\mu,K)\\
  p(j(t)|\zeta,\sigma_{j}^{2}(t)) &\sim& \prod_{i=1}^{N}f(j_{i}(t)) \label{eqPjump}.
\end{eqnarray*}
Here $f$ is the spike and slab prior
\begin{equation}\label{eqSpikeSlab}
  f(j_{i}(t))= \zeta\delta_{0}(j_{i}(t)) + \frac{1-\zeta}{\sqrt{2\pi}\sigma_{j,i}(t)}\exp\left(-\frac{j_{i}(t)^{2}}{2\sigma_{j,i}^{2}(t)}\right)
\end{equation}
with $\delta_{0}$ being a point mass distribution at 0.   The initial time parameters, $\mu$ and $K$ can be chosen based on prior stock return data and will be treated as known values.
\subsection{Prior Distribution of Parameters}\label{sec:Priors}
To allow for more flexible modeling we shall impose prior distributions on the parameters $D,\Gamma,\sigma_{o,i}^{2}$ as well as the jump parameters $\zeta$ and $\sigma_{j,i}^{2}$.   Here we take a commonly used approach of using conjugate prior distributions which facilitate  calculation of conditional maximum a posteriori (MAP) parameter estimates.   These priors will play an essential part in the proofs of convergence for the ECM algorithm presented in Section \ref{Sec:ECM}.

The drift parameter $D$ is modeled as normally distributed with mean $\bar{D}$ and covariance $\sigma_{D}^{2} I$
\begin{equation*}
  D \sim \mathcal{N}(\bar{D}, \sigma_{D}^{2} I),
\end{equation*}
which is conjugate to the multivariate normal distribution given above.
For the covariance matrix prior we use an inverse Wishart prior (which is also conjugate to the multivariate normal) with $\eta > N-1$ degrees of freedom and positive definite scale matrix $W_{o}$
\begin{equation*}
  \Gamma \sim \mathcal{W}^{-1}(W_{o},\eta).
\end{equation*}
In the observation noise variance,$\sigma_{o,i}^{2}$, we impose a inverse gamma distribution with shape parameter $\alpha_{o}>0$ and scale $\beta_{o}>0$
\begin{equation*}
  \sigma_{o,i}^{2} \sim IG(\alpha_{o},\beta_{o}).
\end{equation*}
Finally for the jump parameters $\zeta$ and $\sigma_{j}^{2}$ we use the beta distribution and inverse gamma distribution as priors
\begin{equation*}
  \zeta \sim \textnormal{Beta}(\alpha_{\zeta},\beta_{\zeta})
\end{equation*}
\begin{equation*}
  \sigma_{j,i}^{2}(t) \sim IG(\alpha_{j},\beta_{j}).
\end{equation*}
We assume that $\zeta$ and $\sigma_{j,i}^{2}(t)$ are independent and that the parameters in each of the prior distributions is known.   For each of these priors the hyperparameters may be selected to make them relatively uninformative.

\subsection{Mixture Model Representation}
We may also represent our jump model as mixture model with $2^{TN-N}$ components.   To see this we condition on $Z(1:T)$ and obtain the following
\begin{eqnarray*}
  p(y(t)|x(1:T),\Sigma_{o}(t)) &\sim& \mathcal{N}(\tilde{I}(t)x(t),\Sigma_{o}(t)) \label{eqMix1} \\
  p(x(t)|x(1:t-1),z(2:t),d,\Gamma) &\sim& \mathcal{N}(x(t-1)+d,\Gamma+\textnormal{Diag}(t,z(t))) \\
  p(x(1)) &\sim& \mathcal{N}(\mu,K)\\
  p(z(t)|\zeta) &\sim& \zeta^{TN-T_{J}}(1-\zeta)^{T_{J}} \label{eqMix2}
\end{eqnarray*}
where $T_{J}$ is the total number of jumps
\begin{equation*}
  T_{J}=\sum_{i,t} z_{i}(t)
\end{equation*}
and where $\textnormal{Diag}(t,z(t))$ is the diagonal matrix
\begin{eqnarray*}
\textnormal{Diag}(t,z(t))= \left( \begin{matrix}
               z_{1}(t)\sigma_{j,1}^{2}(t) & 0 & \dots & 0 \\
               0 & \ddots & \ddots & \vdots \\
               \vdots & \ddots  & \ddots & 0 \\
               0 & \dots & 0 &z_{N}(t)\sigma_{j,N}^{2}(t)  \\
             \end{matrix} \right).
\end{eqnarray*}
Here we see that the covariance is time-varying and at time $t$ is equal to $\Gamma + \textnormal{Diag}(t,Z(t))$.   Thus our model is equivalent to a large switching state space model \cite{Ghahramani} with a log-price posterior distribution consisting of $2^{TN-N}$ components.

\subsection{Laplace Prior Approximation}\label{Sec:Laplace}
Recall from the previous section that jump model is equivalent to a switching state space model.   Inference in switching state space models becomes intractable as the number of states increase \cite{Ghahramani}.   For example estimation of the posterior distribution of $X$ given $Y$ involves marginalizing out the $2^{TN-N}$ possible states for $Z$, which is an intractable integral.   Maximum a posteriori (MAP) estimation of $Z$ is also difficult due to the multimodal structure of $p(z)$.

In this section we approximate the distribution of $J$ using a Laplace distribution.   We denote the Laplace distribution for $J$ as $g(j)$
\begin{equation*}
  p(j) \approx g(j|\lambda) \doteq \prod_{i,t}\frac{\lambda_{i}(t)}{2}\exp\left(-\lambda_{i}(t)|j_{i}(t)|\right)
\end{equation*}
where $\lambda_{i}(t)>0$.

There are two advantages to taking this approximation.   First the log-likelihood of a Laplace distribution is concave in its parameter.   This aids in conditional MAP estimation of $J$.  Secondly, the Laplace distribution is desired in that it promotes sparse MAP estimates of $J$ \cite{Seeger1,BoydRobustKalman,AravkinL1Laplace} making it a good approximation to infrequent jumps.    We illustrate this with the following example.
\begin{example}
Suppose that $\kappa$ is a Laplace distributed random variable with parameter $2$ and $\eta=\kappa + q$ where $q$ is independent of $\kappa$ and normally distributed with mean 0 and variance 1.   Suppose $\eta$ is observed to be 0.5.  Then the likelihood of $\kappa$ given $\eta$ is $\mathcal{N}(0.5,1)$ but the posterior of $\kappa$ has its mode at 0 as shown in Figure \ref{fig:LaplacePrior}.

\end{example}
\begin{figure}[h]
    \centering
    \includegraphics[width=5in]{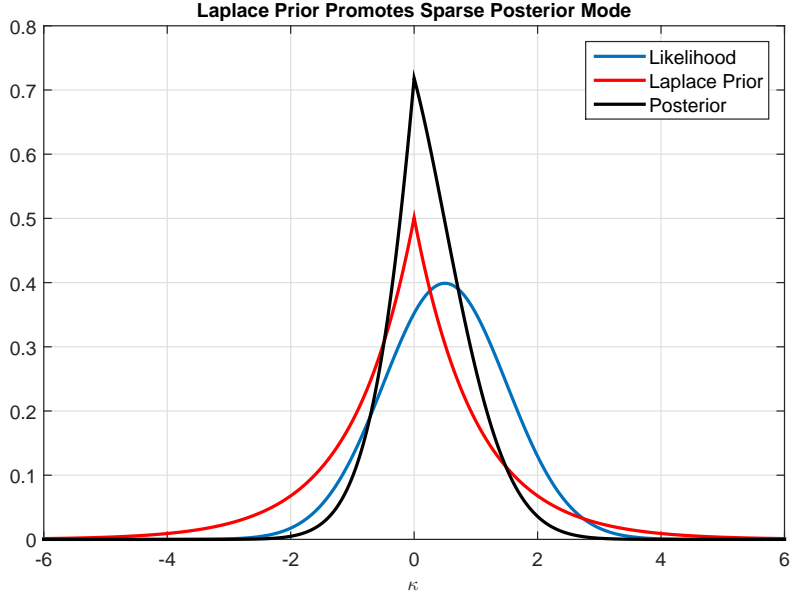}
    \caption{Here we show an example of a Laplace prior promoting a posterior mode at 0.}
    \label{fig:LaplacePrior}
\end{figure}

To make the model more robust we will not assume that each $\lambda_{i}(t)$ is known.   Instead we will estimate $\lambda_{i}(t)$ from the data.   Since the problem of estimating both $J_{i}(t)$ and $\lambda_{i}(t)$ is ill-posed  we regularize it by introducing a  prior distribution on each $\lambda_{i}(t)$ which we denote as $q(\lambda)$.

We wish to design the prior distribution $q$ such that it induces a similar level of sparseness that is induced by the spike and slab prior $f$.
To develop a criterion for designing $q$ we first define a notion of similarity between $g(j|\lambda)$ and $f(j|,\zeta,\sigma_{j}^2)$.
\begin{Def}
Let $V$ be a zero-mean normal random variable with variance $\sigma_{v}^{2}$ and let $J_{1} \sim Laplace(\lambda')$ and $J_{2} \sim SpikeSlab(\zeta',\sigma_{j}^{2'})$ which are independent of $V$.     Define
 \begin{eqnarray}\label{eq:JumpLaplace}
   Y_{1}&=&J_{1}+V \nonumber \\
   Y_{2}&=&J_{2}+V. \nonumber
 \end{eqnarray}
 Then $Laplace(\lambda')$ is \textbf{$\sigma_{v}^{2}$-equivalent} to $SpikeSlab(\zeta',\sigma_{j}^{2'})$ (denoted $\lambda' \sim_{\sigma_{v}^{2}} (\zeta',\sigma_{j}^{2'}) $ ) if
\begin{equation}\label{eqSparseEquiv}
  \E_{p(y_{2}|J_{2}=0)}Pr(J_{2}=0|Y_{2}) = \E_{p(y_{1}|J_{1}=0)}Pr(\bar{J_{1}}=0)
\end{equation}
where $\bar{J_{1}}$ is the mode of $p(j_{1}|Y_{1})$.
\end{Def}
To interpret the above definition assume that a jump has not occurred.   Then $\lambda' \sim_{\sigma_{v}^{2}} (\zeta',\sigma_{j}^{2'}) $ if the probability of falsely declaring a jump under the $Laplace(\lambda')$ model (with MAP criterion) equals the average posterior probability of a jump under the spike and slab prior with parameters $\zeta'$ and $\sigma_{j}^{2'}$.
Here $\sigma_{v}^{2}$ can be interpreted as the squared volatility of the diffusion component of the asset returns.   Note that for each triplet $(\sigma_{v}^{2},\zeta',\sigma_{j}^{2'})$ there is a unique $\lambda'$ such that $\lambda' \sim_{\sigma_{v}^{2}} (\zeta',\sigma_{j}^{2'})$.

Since $(\sigma_{v}^{2},\zeta',\sigma_{j}^{2'})$ are random and unobserved we cannot directly select a $\lambda'$ such that $\lambda' \sim_{\sigma_{v}^{2}} (\zeta',\sigma_{j}^{2'})$.  However the distribution of $(\sigma_{v}^{2},\zeta_{o}',\sigma_{j}^{2'})$  induces a distribution on $\lambda$ through the mapping $\sim_{\sigma_{v}^{2}}$.   The resulting distribution can then be used as a prior  $q(\lambda)$.   The following section presents an example on how to construct a distribution for $\lambda$.
\subsection{Procedure for selecting $q(\lambda)$}\label{Sec:ProcedureQlam}
In this section we outline the method for selecting the distribution $q(\lambda)$ for a special case of when the prior distribution of volatility of each asset is identical.  Suppose  the squared volatility of each asset return is inverse gamma distributed with scale $c$ and shape $\eta$.   Let $\sigma_{v}^{2}$ be distributed as $IG(c,\eta)$ and be statistically independent of $\zeta'$ and $\sigma_{j}^{2}$.

To determine an appropriate prior distribution of $\lambda$ we first obtain samples of $\lambda$, ($\tilde{\lambda}_{1}, \dots, \tilde{\lambda}_{M_{\lambda}}$) by the performing the following steps
  \begin{enumerate}
    \item For $k=1, \dots, M_{\lambda}$
    \item Draw independent samples from the distribution of $(\sigma_{v}^{2'},\zeta',\sigma_{j}^{2'})$.   This is relatively straight forward using standard statistical functions due to the independence assumptions.
    \item Determine a $\lambda'$ such that $\lambda' \sim_{\sigma_{v}^{2'}} (\zeta',\sigma_{j}^{2'})$.   This can be done via Monte Carlo integration as shown below.
    \begin{itemize}
       \item For a large number $L$ draw a sample $v_{1} \dots v_{L}$ from the distribution $\mathcal{N}(0,\sigma_{v}^{2})$.
       \item Compute $P_{i}=Pr(J=0|J+V=v_{i})$, where $J \sim SpikeSlab(\zeta',\sigma_{j}^{2'})$.   The value of $P_{i}$ is
           \begin{equation*}
             \frac{ \frac{\zeta'}{\sqrt{\sigma_{v}^{2'}}}\exp(-v_{i}^{2}/(2\sigma_{v}^{2}))  }{\frac{\zeta'}{\sqrt{\sigma_{v}^{2'}}}\exp(-v_{i}^{2}/(2\sigma_{v}^{2'})) + \frac{1-\zeta'}{\sqrt{\sigma_{v}^{2'}+\sigma_{j}^{2'}}}\exp(-v_{i}^{2}/(2( \sigma_{v}^{2'}+\sigma_{j}^{2'})))}.
           \end{equation*}
       \item Compute the simulated empirical mean $\bar{P}=\frac{1}{L}\sum_{i=1}^{L} P_{i}$.
       \item Choose $\lambda'$ such that \eqref{eqSparseEquiv} is satisfied with $\E_{p(y_{2}|J_{2}=0)}Pr(J_{2}=0|Y_{2})$ approximated as $\bar{P}$.  This value is given below
            \begin{equation*}
             \lambda' =\frac{\textnormal{erf}^{-1}(\bar{P})\sqrt{2\sigma_{v}^{2'}}}{\sigma_{v}^{2'}}
            \end{equation*}
           where $\textnormal{erf}^{-1}()$ is the inverse error function.
   \end{itemize}

    \item Set $\tilde{\lambda}_{k}=\lambda'$
    \item Goto step 1

  \end{enumerate}
  Examples of histograms of samples obtained using the above procedures are shown in Figures \ref{fig:qLamVsSigV} - \ref{fig:qLamVsZeta}.
  Once we obtain samples of $\lambda$ we fit a smooth distribution to the sampled data.  Since the gamma distribution is a conjugate prior to the Laplace distribution a gamma distribution is a convenient choice for $q(\lambda)$.   Furthermore examination of Figures \ref{fig:qLamVsSigV} -\ref{fig:qLamVsZeta} indicate that a gamma distribution is a reasonable approximation.   Thus we choose
  \begin{equation*}
    q(\lambda) = \frac{\beta_{\lambda}^{\alpha_{\lambda}}}{\Gamma_{f}(\alpha_{\lambda})} \lambda^{\alpha_{\lambda}-1}\exp\left(-\lambda \beta_{\lambda}\right)
  \end{equation*}
  where $\Gamma_{f}()$ is the gamma function.  Here $\alpha_{\lambda}$ and $\beta_{\lambda}$ can be selected using maximum likelihood or method of moments.

  Since $q(\lambda)$ develops a singularity near zero for large values of $\beta_{\lambda}$ we shall impose a prior of $\lambda^{-1}$ rather than $\lambda$.    We denote this prior as $q_{inv}(\lambda^{-1})$.   Since $\lambda$ is gamma distributed with shape $\alpha_{\lambda}$ and rate $\beta_{\lambda}$ it follows that $q_{inv}(\lambda^{-1})$ is the inverse gamma distribution with shape $\alpha_{\lambda}$ and scale $\beta_{\lambda}$
  \begin{equation*}
    q_{inv}(\lambda^{-1})=\frac{\beta_{\lambda}^{\alpha_{\lambda}}}{\Gamma_{f}(\alpha_{\lambda})}(\lambda^{-1})^{-\alpha_{\lambda}-1}\exp\left(-\frac{\beta_{\lambda}}{\lambda^{-1}}\right).
  \end{equation*}
\begin{figure}[h]
    \centering
    \includegraphics[width=5in]{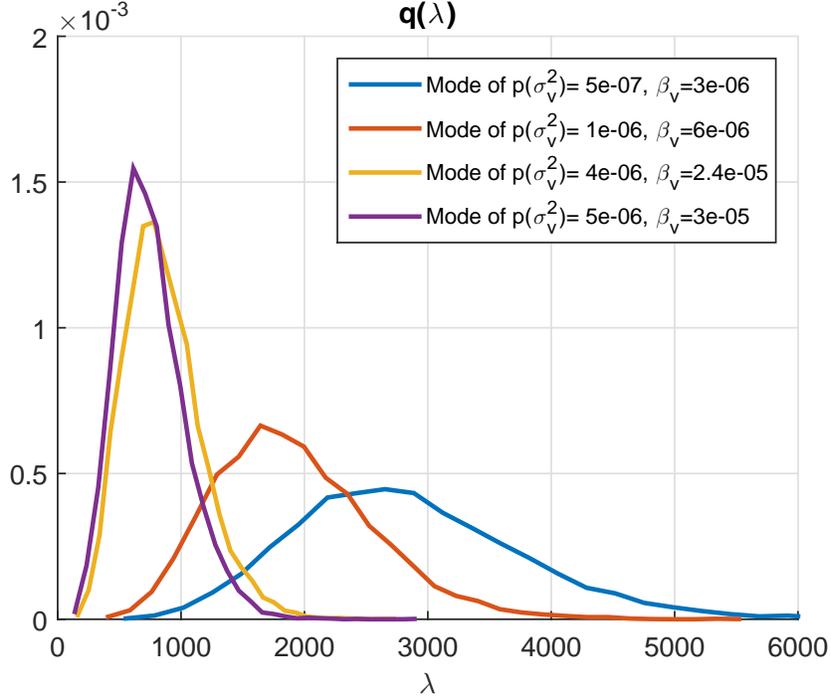}
    \caption{Normalized histograms of $\lambda$ samples.  In all experiments $\sigma_{j}^{2} \sim IG(10,0.0011)$,$\zeta \sim \textnormal{Beta}(5,1.0201)$, $\sigma_{v}^{2} \sim IG(5,\beta_{v})$.}
    \label{fig:qLamVsSigV}
\end{figure}
\begin{figure}[h]
    \centering
    \includegraphics[width=5in]{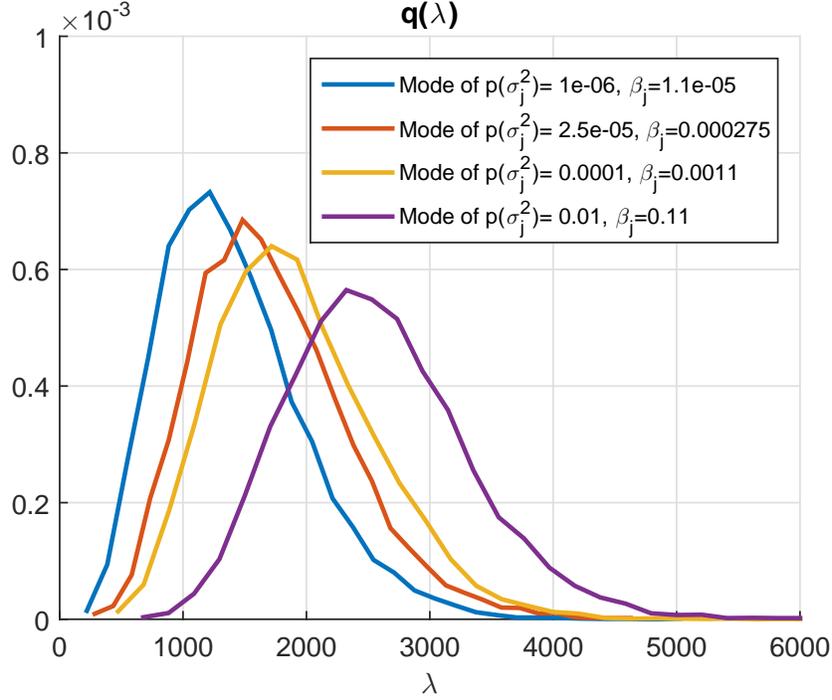}
    \caption{Normalized histograms of $\lambda$ samples.  In all experiments $\sigma_{j}^{2} \sim IG(10,\beta_{j})$,$\zeta \sim \textnormal{Beta}(5,1.0201)$, $\sigma_{v}^{2} \sim IG(5,6e-6)$.}
    \label{fig:qLamVsJump}
\end{figure}
\begin{figure}[h]
    \centering
    \includegraphics[width=5in]{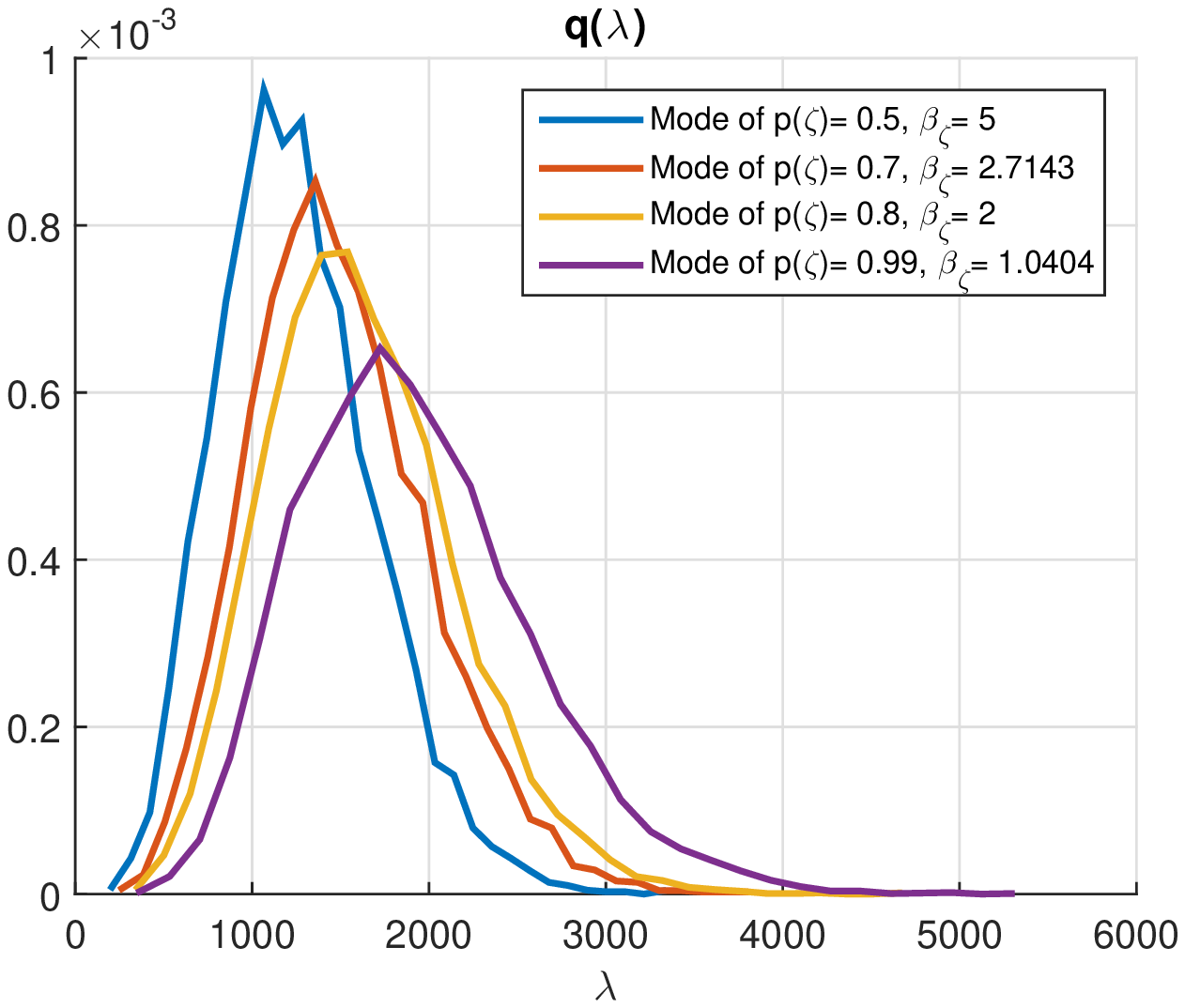}
    \caption{Normalized histograms of $\lambda$ samples.  In all experiments $\sigma_{j}^{2} \sim IG(10,0.0011)$,$\zeta \sim \textnormal{Beta}(5,\beta_{\zeta})$, $\sigma_{v}^{2} \sim IG(5,6e-6)$.}
    \label{fig:qLamVsZeta}
\end{figure}

\section{KECM Approach to estimation of $\Gamma$}
\label{Sec:ECM}
Maximum a posteriori (MAP) estimation of $\Gamma$ with Kalman ECM (KECM) techniques is investigated in this section.  The first ECM approach is an approximate technique where the prior distribution on the jumps is modeled as a Laplace distribution.  The advantage of this approximation is that the conditional  maximization steps in the ECM approach result in global (conditional) optimal solutions can be obtained.
The disadvantage is that we are approximating the true spike and slab jump model.   The second approach uses the spike and slab model for jumps, which is a true representation of the model presented in Section \ref{Sec:ModelHigh}.   However we will see that using the spike and slab jump model results in a non-concave optimization problem in the conditional M-step for $J$.

\subsection{KECM algorithm for Laplace Distribution}
First we consider a KECM approach to estimating $\Gamma$ when $J_{i}$ is approximated by a Laplace distributed random variable.  We define
\begin{equation*}
\Theta = [\Theta_{1},\Theta_{2},\Theta_{3},\Theta_{4},\Theta_{5}]
\end{equation*}
where
\begin{eqnarray*}
\Theta_{1} &=& D  \\
\Theta_{2} &=& \Gamma \\
\Theta_{3} &=& \sigma_{o,i}^{2}, 1 \le i \le N \\
\Theta_{4} &=& J(2:T) \\
\Theta_{5} &=& \lamSetInv
\end{eqnarray*}
as our vector of unknown parameters and $X(1:T)$ as the latent variables.

The KECM approach is an iterative algorithm that can be applied to the following problem
\begin{equation*}
  \Theta^{*} = \arg\max_{\theta} L(\theta)
\end{equation*}
where $L(\theta)$ is the log posterior of $\Theta$.   In the KECM algorithm we iterate over E-steps and conditional M-steps to arrive at an estimate of $\Theta$.

The E-step in the KECM algorithm involves computing the expected value of
\begin{equation*}
  \log p(X(1:T),y(1:T)|\theta)p(\theta)
\end{equation*}
with respect to $p(x(1:T)|y,\Theta^{(k)})$
\begin{equation*}
  \mathcal{G}(\theta,\Theta^{(k)}) = \mathbb{E}_{p(x|y,\Theta^{(k)})}\log p(X(1:T),y(1:T)|\theta)+\log(p(\theta))
\end{equation*}
where $\Theta^{(k)}$ is an estimate of $\Theta$ at the $k^{th}$ iteration and where $p(\theta)$ is the prior distribution of parameters
\begin{equation*}
p(\theta)=p(\theta_{1})p(\theta_{2})p(\theta_{3})g(\theta_{4},|\lambda)q_{inv}(\lambda^{-1}).
\end{equation*}
Here the complete log-likelihood is
\begin{eqnarray*}
  \log p(x,y|\theta) &=& -0.5\sum_{t=1}^{T}\log(|\Sigma_{o}(t)|) -\frac{1}{2}\sum_{t=1}^{T}||y(t)-\tilde{I}(t)x(t)||_{\textnormal{diag}(\Sigma_{o}(t)^{-1}),\ell_{2}}^{2} \nonumber \\
   && - \frac{T-1}{2}\log(|\Gamma|) \nonumber \\
   && - \frac{1}{2}\sum_{t=2}^{T}r(t)^{T}\Gamma^{-1}r(t) \nonumber \\
   && + const
\end{eqnarray*}
where
\begin{equation*}
   r(t)=x(t)-x(t-1)-d-j(t).
\end{equation*}
and where
\begin{equation*}
||q||_{\beta,\ell_{2}}^{2} = \sum_{i}\beta_{i}q_{i}^{2}.
\end{equation*}

It is well known that the function $\mathcal{G}(\theta,\Theta^{(k)})$ serves as a lower bound to $\log p(\theta,y)$ and that
$\log p(\Theta^{(k)},y) = \mathcal{G}(\Theta^{(k)},\Theta^{(k)})$ \cite{EMRubin}.

The EM approach prescribes a joint maximization of $\mathcal{G}(\theta,\Theta^{(k)})$  with respect to $\theta$.   This is difficult due to the coupling of variables and the non-concavity of the problem.   Conditional maximization of each parameter in turn is more tractable.  Thus we apply conditional maximization as in the ECM \cite{ECMRubin} algorithm.   The conditional M-steps involves a coordinate-wise maximization of $\mathcal{G}$.  Here the conditional M-steps are
\begin{eqnarray}
  \Theta_{1}^{(k+1)} = \arg\max_{\theta_{1}} \mathcal{G}\left(\left[\theta_{1},\Theta_{2}^{(k)},\Theta_{3}^{(k)},\Theta_{4}^{(k)},\Theta_{5}^{(k)} \right],\Theta^{(k)}\right) \nonumber \\
  \Theta_{2}^{(k+1)} = \arg\max_{\theta_{2}} \mathcal{G}\left(\left[\Theta_{1}^{(k+1)},\theta_{2},\Theta_{3}^{(k)},\Theta_{4}^{(k)},\Theta_{5}^{(k)}\right],\Theta^{(k)}\right) \nonumber \\
  \Theta_{3}^{(k+1)} = \arg\max_{\theta_{3}} \mathcal{G}\left(\left[\Theta_{1}^{(k+1)},\Theta_{2}^{(k+1)},\theta_{3},\Theta_{4}^{(k)},\Theta_{5}^{(k)}\right],\Theta^{(k)}\right) \nonumber \\
  \Theta_{4}^{(k+1)} = \arg\max_{\theta_{4}} \mathcal{G}\left(\left[\Theta_{1}^{(k+1)},\Theta_{2}^{(k+1)},\Theta_{3}^{(k+1)},\theta_{4},\Theta_{5}^{(k)}\right],\Theta^{(k)}\right) \nonumber \\
  \Theta_{5}^{(k+1)} = \arg\max_{\theta_{5}} \mathcal{G}\left(\left[\Theta_{1}^{(k+1)},\Theta_{2}^{(k+1)},\Theta_{3}^{(k+1)},\Theta_{4}^{(k+1)},\theta_{5}\right],\Theta^{(k)}\right).  \label{eqTheta3}
\end{eqnarray}
Each of these problems can be readily solved as we will show later.
\subsubsection{E-step of KECM}
The posterior $p(x|y,\Theta^{(k)})$ needed to perform the E-step is normal and can be computed using a Kalman smoother \cite{Shumway}.
By normality and the Markov property the posterior is completely defined by the following posterior moments for $m=T$
\begin{equation*}
  \bar{X}(t|m) \doteq \E(X(t)|y(1:m))
\end{equation*}
\begin{equation*}
  P(t|m) \doteq \cov(X(t),X(t)|y(1:m))
\end{equation*}
\begin{equation*}
  P(t,t-1|m) \doteq \cov(X(t),X(t-1)|y(1:m))
\end{equation*}
where $\cov(:,:)$ refers to the covariance function.   Equations for these quantities are derived in \cite{ShumwayBook} and are stated in Appendix \ref{secAppendixSmooth}.
The expected value of log-posterior distribution with respect to the posterior of $X(1:T)$ can be shown to be
\begin{eqnarray}\label{eqPosteriorECM}
  \mathcal{G}(\theta,\Theta^{(k)})&=& \mathbb{E}_{p(x|y,\Theta^{(k)})}\log p(X(1:T),y(1:T)|\theta)+\log(p(\theta)) \nonumber \\
  &=&-\frac{T-1}{2}\log(|\Gamma|) - \frac{1}{2}\textnormal{tr}(\Gamma^{-1}(C-B-B^{T}+A)) \nonumber \\
   & & -0.5\sum_{t=1}^{T}\log(|\Sigma_{o}(t)|)\nonumber \\
   &&  -\frac{1}{2}\sum_{t=1}^{T}||y(t)-\tilde{I}(t)\bar{X}(t)||_{\textnormal{diag}(\Sigma_{o}(t)^{-1}),\ell_{2}}^{2} + \textnormal{tr}(P(t|T)\tilde{I}(t)^{T}\Sigma_{o}(t)^{-1}\tilde{I}(t)) \nonumber \\
   && + \log(p(\theta)) + const
\end{eqnarray}
where
\begin{equation*}
  A=\sum_{t=2}^{T} \left(P(t-1|T) + \bar{X}(t-1|T)\bar{X}(t-1|T)^{T} \right)
\end{equation*}
\begin{equation*}
  B=\sum_{t=2}^{T} \left(P(t,t-1|T) + (\bar{X}(t|T)-D^{(k)}-J^{(k)}(t))\bar{X}(t-1|T)^{T} \right)
\end{equation*}
\begin{equation*}
  C=\sum_{t=2}^{T} \left(P(t|T) + (\bar{X}(t|T)-D^{(k)}-J^{(k)}(t))(\bar{X}(t|T)-D^{(k)}-J^{(k)}(t))^{T} \right).
\end{equation*}
These equations are derived in Appendix \ref{secAppendixDerive}.   For notational convenience the dependence of $P(t|m)$ and $P(t,t-1|m)$ on the iteration number has been dropped.

\subsubsection{Conditional M-steps of KECM}
For the conditional M-step it can be shown using standard conjugate prior relationships \cite{FinkConjPrior} that
\begin{equation}\label{driftEM}
  D^{(k+1)} = F\left ( \frac{1}{\sigma_{D}^2}\bar{D}+ \Gamma^{(k)^{-1}}\sum_{t=2}^{T}\bar{X}(t|T)-\bar{X}(t-1|T)-J^{(k)}(t) \right)
\end{equation}
and
\begin{equation}\label{GammaEM}
  \Gamma^{(k+1)} = \frac{1}{T-1+\eta} \left(A+C^{(k)}-B^{(k)}-B^{(k)T}\right) +\frac{1}{T-1+\eta}W
\end{equation}
where
\begin{equation*}
  F = \left((T-1)\Gamma^{(k)^{-1}}+\sigma_{D}^{-2}I \right)^{-1}
\end{equation*}
\begin{equation*}
  B^{(k)}=\sum_{t=2}^{T} \left(P(t,t-1|T) + (\bar{X}(t|T)-D^{(k+1)}-J(t)^{(k)})\bar{X}(t-1|T)^{T} \right)
\end{equation*}
and
\begin{eqnarray*}
  C^{(k)}&=& \sum_{t=2}^{T} P(t|T) \\
   &&  + \sum_{t=2}^{T} (\bar{X}(t|T)-D^{(k+1)}-J(t)^{(k)})(\bar{X}(t|T)-D^{(k+1)}-J(t)^{(k)})^{T}.
\end{eqnarray*}

The conditional M-step for the observation noise variance is
\begin{equation}\label{obsEM}
  \sigma_{o,i}^{2,(k+1)} = \frac{2\beta_{o} + \sum_{t\in \mathcal{T}_{i}}(y(t)-\tilde{I}(t)\bar{X}(t|T))_{\eta(i,t)}^{2} + (P(t|T))_{i,i} }{2\alpha_{o}+2+M_{i}}.
\end{equation}
Here $\mathcal{T}_{i}$ is the set of times where the price of asset $i$ is observed and $M_{i}$ is the total number of prices observed for asset $i$.
The subscript $\eta(i,t)$ is the row number of $\tilde{I}(t)$ such that $\tilde{I}(t)_{\eta(i,t),i}=1$.

For each conditional M-step $P(t,T)$, $P(t,t-1|T)$ and $\bar{X}(t|T)$ are evaluated with respect to $p(X(1:T)|Y,\Theta^{(k)})$.

To compute the conditional M-step for $J$ we denote
\begin{eqnarray*}
Q(j) &\doteq& \mathcal{G}([\Gamma^{(k+1)},D^{(k+1)},\sigSet,j,\lamSetInv^{(k)}],\Theta^{(k)}).
\end{eqnarray*}
Then up to a constant not depending on $j$
\begin{eqnarray*} \
Q(j) &=& - \frac{1}{2} \sum_{t=2}^{T}(\bar{X}(t|T)-j(t)-D^{(k+1)})^{T}( \Gamma^{(k+1)})^{-1}(\bar{X}(t|T)-j(t)-D^{(k+1)})\nonumber \\
      && + \sum_{t=2}^{T} (\bar{X}(t|T)-j(t)-D^{(k+1)})^{T}( \Gamma^{(k+1)})^{-1}\bar{X}(t-1|T) \nonumber \\
      &&+ \log(g(j(2:T)|\lamSet^{(k)})) + const_{1}\\
     &=&  - \frac{1}{2}\sum_{t=2}^{T} (\bar{X}(t|T)-j(t)-D^{(k+1)})^{T}( \Gamma^{(k+1)})^{-1}(\bar{X}(t|T)-j(t)-D^{(k+1)})\nonumber \\
      && + \sum_{t=2}^{T}(\bar{X}(t|T)-j(t)-D^{(k+1)})^{T}( \Gamma^{(k+1)})^{-1}\bar{X}(t-1|T) \nonumber \\
      &&- \sum_{t=2}^{T}\EllnW{j(t)}{\lambda(t)}{1} + const_{2}.
\end{eqnarray*}
where $\EllnW{j(t)}{\lambda(t)}{1} = \sum_{n=1}^{N}\lambda_{n}(t)|j_{n}(t)|$.
By rearranging terms we can express $Q(j)$ as a quadratic function of $j$
\begin{eqnarray*} \
Q(j) &=& - \frac{1}{2} \sum_{t=2}^{T}j(t)^{T}( \Gamma^{(k+1)})^{-1}j(t)\nonumber \\
      && + \sum_{t=2}^{T}(\bar{X}(t|T)-D^{(j+1)}-\bar{X}(t-1|T))^{T}(\Gamma^{(k+1)})^{-1}j(t) \nonumber \\
      &&- \sum_{t=2}^{T}\EllnW{j(t)}{\lambda(t)}{1} + const_{3}.
\end{eqnarray*}

Referring to equation \eqref{eqTheta3} we see that $J^{(k+1)}(t)$ is the solution of the following $\ell_{1}$ penalized quadratic program
\begin{align}\label{eqJumpProb}
J^{(k+1)}(t) =\arg\min_{j} \frac{1}{2}j^{T}( \Gamma^{(k+1)})^{-1} j - j^{T}(\Gamma^{(k+1)})^{-1}\Delta^{(k+1)} + \EllnW{j(t)}{\lambda(t)}{1}
\end{align}
where
\begin{equation} \label{eqJumpLaplaceEM}
  \Delta^{(k)}(t)=\bar{X}(t|T)-D^{(k)}-\bar{X}(t-1|T).
\end{equation}
This problem can be solved with a variety of fast algorithms such as ADMM \cite{ADMM} and FISTA \cite{FISTA}.

Now we determine $\lamSetInv$ which depends only on $q_{inv}(\lambda^{-1})$ and $p(j|\lambda)$.  Using conjugate prior relationships
we have $p(\lambda_{i}(t)^{-1}|j_{i}(t))$ is inverse gamma distributed with shape $\alpha_{\lambda}+1$ and scale $\beta_{\lambda}+ |j_{i}(t)|$.   Thus the conditional MAP estimate is
\begin{equation}\label{eqLamInvUp}
  \lambda_{i}(t)^{-1} = \frac{|J^{(k+1)}_{i}(t)|+\beta_{\lambda}}{\alpha_{\lambda}+2}.
\end{equation}
which implies that
\begin{equation}\label{eqLamUp}
  \lambda_{i}(t) = \frac{\alpha_{\lambda}+2}{|J^{(k+1)}_{i}(t)|+\beta_{\lambda}}.
\end{equation}

An outline of the KECM algorithm for Laplace jump models is given below.
\begin{algorithm}[H]
\caption{KECM Algorithm for estimation of $\Gamma$ under Laplace Prior}
\label{Alg:ECM}
\begin{algorithmic}
\STATE  \textbf{Initialize: }$\Theta^{(0)},k=0$
\WHILE{ not converged}
\STATE Compute $\bar{X}(t|T), P(t|T),P(t,t-1|T)$ using Kalman smoothing equations for $\Theta^{(k)}$  using equations \eqref{eq:ForwardEq}-\eqref{eq:BackwardRecursion}
\STATE Compute $D^{(k+1)},\Gamma^{(k+1)}$, and $\sigma_{o,i}^{2,(k+1)}$ using equations \eqref{driftEM},\eqref{GammaEM}, and \eqref{obsEM} respectively
\STATE Compute $J^{(k+1)}$ by solving \eqref{eqJumpLaplaceEM}
\STATE Compute $\lamSet$ by solving \eqref{eqLamUp}
\STATE $k=k+1$
\ENDWHILE
\end{algorithmic}
\end{algorithm}
Convergence results for this algorithm are given in Appendix \ref{Sec:ConvergeECM}.

\begin{remark}
\textnormal{
Since of value of $\lamSet$ changes with each iteration we see that we effectively reweight the $\ell_{1}$ penalty in \eqref{eqJumpProb} after each iteration.   Reweighting of the $\ell_{1}$ norm has been proposed in several papers and has been shown to have improved performance in compressive sensing problems versus a fixed set of weights \cite{CandesReweightedL1}.}
\end{remark}

\subsection{KECM approach for the Spike and Slab Jump Prior}
Now we present a KECM for the spike and slab jump prior.   As with the Laplace prior we treat $X$ as a latent variable.   Let us denote the unknown parameters as $\Phi$ where
\begin{eqnarray}
\Phi_{1} &=& D \nonumber \\
\Phi_{2} &=& \Gamma \nonumber \\
\Phi_{3} &=& \sigma_{o,i}^{2}\nonumber \\
\Phi_{4} &=& Z(2:T),\tilde{J}(2:T) \nonumber \\
\Phi_{5} &=& \zeta \nonumber \\
\Phi_{6} &=& \left\{\sigma_{j,i}^{2}(t)\right\}_{i=1,\dots,N,t=1,\dots,T}.  \nonumber
\end{eqnarray}
Here we allow for distinct $\sigma_{j}^{2}$ values for each time and asset.

The E-step as well as the conditional M-steps for $\Phi_{1},\Phi_{2},\Phi_{3}$ are identical to the KECM algorithm for Laplace priors.   The differences for this section are in the conditional M-steps for $J$, $\zeta$ and $\sigma_{j}^{2}$.

First we address the conditional M-step for $J^{(k+1)}$.  Here we need to solve
\begin{eqnarray}\label{eqJumpProbSpike}
[J^{(k+1)}(t),Z^{(k+1)}(t),\tilde{J}^{(k+1)}(t)] &=&\arg\min_{j,z,\tilde{j}} \frac{1}{2}j^{T}( \Gamma^{(k+1)})^{-1} j\nonumber \\
&& \;\;\;\;\;  - j^{T}(\Gamma^{(k+1)})^{-1}\Delta^{(k+1)} - \sum_{i=1}^{N}\log(f(\tilde{j}_{i},z_{i})) \nonumber \\
&& \textnormal{ s.t.  }j_{i} =\tilde{j}_{i}z_{i}
\end{eqnarray}
where
\begin{equation}\label{eq:SpikeSlabPenalty}
  \log f(\tilde{j}_{i},z_{i}) =  \log\left(\zeta 1_{z_{i}=0} +(1-\zeta) 1_{z_{i} = 1}\right)+ \log\left(\frac{1}{\sqrt{2\pi\sigma_{j,i}^{2}}}\exp\left(-\frac{\tilde{j}_{i}^{2}}{2 \sigma_{j,i}^{2}}\right)\right).
\end{equation}
Here we dropped the notation for time dependence.  When restricted to $j_{i} =\tilde{j}_{i}z_{i}$, $-\log(\tilde{j}_{i},z_{i})$ induces a penalty on $j_{i}$.   which is a weighted sum of an $\ell_{0}$ and squared $\ell^{2}$ norm.   A plot of this penalty is shown in Figure \ref{fig:SpikeSlabPen}.
\begin{figure}[h]
    \centering
    \includegraphics[width=5in]{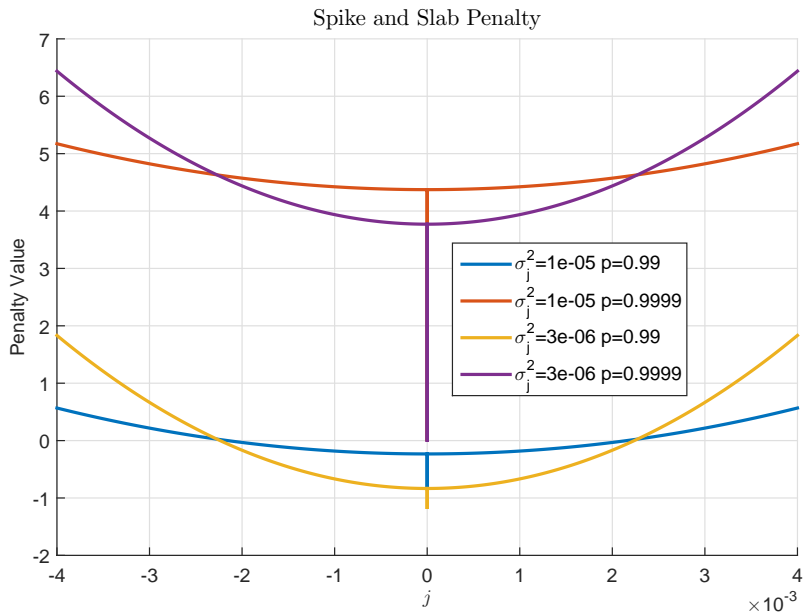}
    \caption{Spike and slab penalty function for various parameter values.   Here we see that the penalty is a weighted sum of $\ell_{0}$ and squared $\ell_{2}$ norms.}
    \label{fig:SpikeSlabPen}
\end{figure}

The term $-\log(\tilde{j}_{i},z_{i})$ is non-convex and complicates the conditional M-step \eqref{eqJumpProbSpike}.   Hence we seek an approximate maximization through coordinate descent.  Here we divide the problem into tractable 1-dimensional optimization problems with respect to one asset at a time.  The method and equations for implementing coordinate descent are derived in Appendix \ref{SecMCMCAppendixJump} and described below.   For ease of notation we drop the notation denoting dependence on $k$.

Let us define the following conditional mean and variance
\begin{equation}\label{eq:aSpike}
  a(i)=\Delta_{i}(t)+\Gamma_{i,-i}\Gamma_{-i,-i}^{-1}(j_{-i}(t)-\Delta_{-i}(t))
\end{equation}
and
\begin{equation}\label{eq:bSpike}
  b^{2}(i)=\Gamma_{i,i}-\Gamma_{i,-i}\Gamma_{-i,-i}^{-1}\Gamma_{-i,i}
\end{equation}
where the subscript $-i$ is to be interpreted as all indices except $i$.
Then the following rule determines the MAP optimal value of $z_{i}(t)$ conditioned on $j_{-i}(t)$
\begin{eqnarray}\label{eq:SpikeZupdate}
  z_{i|-i}(t) = \begin{cases} 0 \mbox { if } \frac{\zeta}{1-\zeta} \mathcal{N}(0,a(i),b^{2}(i)) > \mathcal{N}(0,a(i),b^{2}(i)+\sigma_{j,i}^{2}(t))  \\
  1 \mbox{ else }\end{cases}
\end{eqnarray}
where $\mathcal{N}(0,a(i),b^{2}(i))$ is the normal PDF with mean $a(i)$ and variance $b^{2}(i)$ evaluated at 0.   An optimal value of $\tilde{J}_{i|-i}(t)$ is then given as
\begin{eqnarray}\label{eq:SlabJupdate}
  \tilde{J}_{i|-i}(t)=\begin{cases} \frac{a}{1+b^{2}\sigma_{j,i}^{-2}(t)} \mbox { if } z_{i|-i}(t) \ne 0 \\
  0  \mbox{ else }
  \end{cases}.
\end{eqnarray}

The mapping defined by equations \eqref{eq:SpikeZupdate} and \eqref{eq:SlabJupdate} is a combination of a thresholding step followed by a shrinkage operation
\begin{eqnarray}\label{eq:SpikeUpdate}
  J_{i|-i}(t)&=&SpikeSlabShrink(a,b^{2}) \nonumber \\
   &\doteq& \begin{cases} 0 &\mbox { if } \frac{\zeta \mathcal{N}(0,a(i),b^{2}(i))}{(1-\zeta)\mathcal{N}(0,a(i),b^{2}(i)+\sigma_{j,i}^{2}(t))} >  1 \\
  \frac{a(i)}{1+b^{2}(i)\sigma_{j,i}^{-2}(t)} & \mbox{ else }
  \end{cases}.
\end{eqnarray}
This spike and slab shrinkage is illustrated in Figure \ref{fig:SpikeSlabShrink}.   As the plots indicate the shrinkage is discontinuous and large values are shrunk more than smaller values.
\begin{figure}[h]
    \centering
    \includegraphics[width=5in]{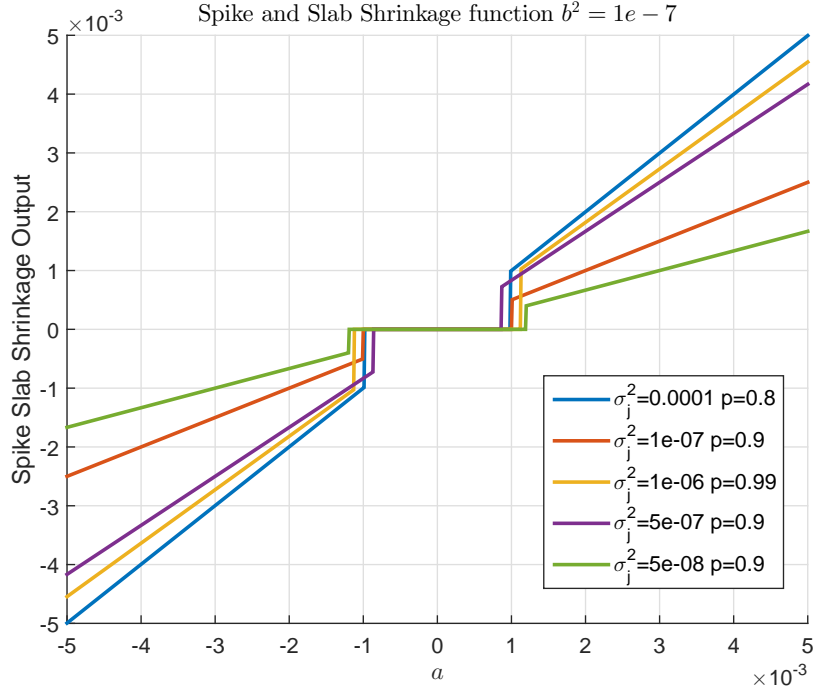}
    \caption{Spike and slab shrinkage function for various parameter values}
    \label{fig:SpikeSlabShrink}
\end{figure}

Equation \eqref{eq:SpikeUpdate} is cycled through all $i=1 \dots N$.  Multiple cycles may also be performed to obtain an improved estimate of $J$.  A summary of the algorithm for the conditional M-step for $J$ is given below in Algorithm \ref{Alg:spikeJ}.
\begin{algorithm}[H]
\caption{Coordinate Descent for Determination of $Z^{(k+1)}(t)$,$\tilde{J}^{(k+1)}(t)$, and $J^{(k+1)}(t)$}
\label{Alg:spikeJ}
\begin{algorithmic}
\STATE  \textbf{Initialize: }Set $J^{(k+1)}(t)=J^{(k)}(t)$, $it$=0, $L>0$
\WHILE{ $it \le L$}
    \STATE $it=it+1$
    \STATE $i=0$
    \WHILE{$i < N$}
        \STATE $i=i+1$
        \STATE Compute $Z_{i}^{(k+1)}(t)$ using equations \eqref{eq:aSpike}, \eqref{eq:bSpike}, and \eqref{eq:SpikeZupdate}
        \STATE Compute $\tilde{J}_{i}^{(k+1)}(t)$ using equations \eqref{eq:aSpike}, \eqref{eq:bSpike}, and \eqref{eq:SlabJupdate}
        \STATE Set $J^{(k+1)}_{i}(t) = Z_{i}^{(k+1)}(t)\tilde{J}_{i}^{(k+1)}(t)$
    \ENDWHILE
\ENDWHILE
\STATE return $J^{(k+1)}(t)$
\end{algorithmic}
\end{algorithm}
Although this method is not guaranteed to solve \eqref{eqJumpProbSpike} it will not increase the value of the objective function compared with $J^{k}(t)$.

Once $J^{(k+1)}$ is obtained, values for $\zeta^{(k+1)}$ and $\sigma_{j}^{2,(k+1)}$ are easily computed through conjugate prior relationships.  First let $N_{Z}$ be number of zero values in $J(2:T)^{(k+1)}$. Then by conjugate prior relationships the conditional M-steps for $\zeta$ and $\sigma_{j}^{2}$ are
\begin{equation}\label{eq:UpdateZeta}
  \zeta^{(k+1)} = \frac{\alpha_{\zeta}+N_{Z}}{N(T-1)+\beta_{\zeta}+\alpha_{\zeta}}
\end{equation}
and
\begin{equation}\label{eq:UpdateSigJ}
  \sigma_{j,i}^{2,(k+1)}(t) = \frac{\beta_{j}+0.5(J_{i}(t))^{2}}{\alpha_{j}+1+0.5(Z_{i}(t))}.
\end{equation}

The KECM algorithm for spike and slab models is summarized in Algorithm \ref{Alg:ECMspike}.
\begin{algorithm}[H]
\caption{KECM Algorithm for estimation of $\Gamma$ under Spike and Slab Prior}
\label{Alg:ECMspike}
\begin{algorithmic}
\STATE  \textbf{Initialize: }$\Phi^{(0)},k=0$
\WHILE{ not converged}
\STATE Compute $\bar{X}(t|T), P(t|T),P(t,t-1|T)$ using Kalman smoothing equations for $\Theta^{(k)}$  using equations \eqref{eq:ForwardEq}-\eqref{eq:BackwardRecursion}
\STATE Compute $D^{(k+1)},\Gamma^{(k+1)}$, and $\sigma_{o,i}^{2,(k+1)}$ using equations \eqref{driftEM}, \eqref{GammaEM}, and \eqref{obsEM} respectively
\STATE For all $t$, compute $\tilde{J}^{(k+1)}(t)$, $Z^{(k+1)}(t)$ using Algorithm \ref{Alg:spikeJ}
\STATE Set $J^{(k+1)}_{i}(t) = Z_{i}^{(k+1)}(t)\tilde{J}_{i}^{(k+1)}(t)$
\STATE Compute $\zeta^{(k+1)}$ using equation \eqref{eq:UpdateZeta}
\STATE Compute $\sigma_{j,i}^{2,(k+1)}(t)$ using equation \eqref{eq:UpdateSigJ} for all $i,t$
\STATE $k=k+1$
\ENDWHILE
\end{algorithmic}
\end{algorithm}

Note that although $J(t)$ is only approximately maximized in each conditional M-step this is still an ECM algorithm.    To see this we can simply redefine $\Phi$ as
\begin{equation*}
  \left[D,\Gamma,\sigma_{o}^{2},J_{1}(2),\dots,J_{N}(2),\dots,J_{1}(T),\dots,J_{N}(T),\zeta,\sigma_{j}^{2}\right].
\end{equation*}
Then the above algorithm is an ECM algorithm for the redefined parameter vector.   The convergence of Algorithm \ref{Alg:ECMspike} is similar to the proof of the convergence of Algorithm \ref{Alg:ECM} in Appendix \ref{Sec:ConvergeECM}.

\begin{remark}
\textnormal{ A comparison of the spike and slab shrinkage function with the shrinkage function of the $b^{2}-equivalent$ Laplace prior is shown in Figure \ref{fig:SpikeSlabShrinkCompare}.   The Laplace shrinkage function (with parameter $\lambda$)  is defined as
\begin{eqnarray*}
  LaplaceShrink(a,b^{2}) &\doteq& \begin{cases} a-\lambda b^{2} &\mbox { if } a>\lambda b^{2}  \\
  a+\lambda b^{2} & \mbox{ if } a<-\lambda b^{2} \\
  0 &\mbox{ else }
  \end{cases}.
\end{eqnarray*}
The graphs illustrate advantages and disadvantages of the Laplace prior.   One notable disadvantage is that for large $\sigma_{j}^{2}$ the Laplace prior has a large bias relative to spike and slab priors.   However for small $\sigma_{j}^{2}$ and large values of $a$ we see that the Laplace prior is less biased than the spike and slab.   This can be attributed to the quadratic penalty induced by the spike and slab prior which penalizes large jumps more heavily than the Laplace prior.}
\end{remark}
\begin{figure}[h]
    \centering
    \includegraphics[width=5in]{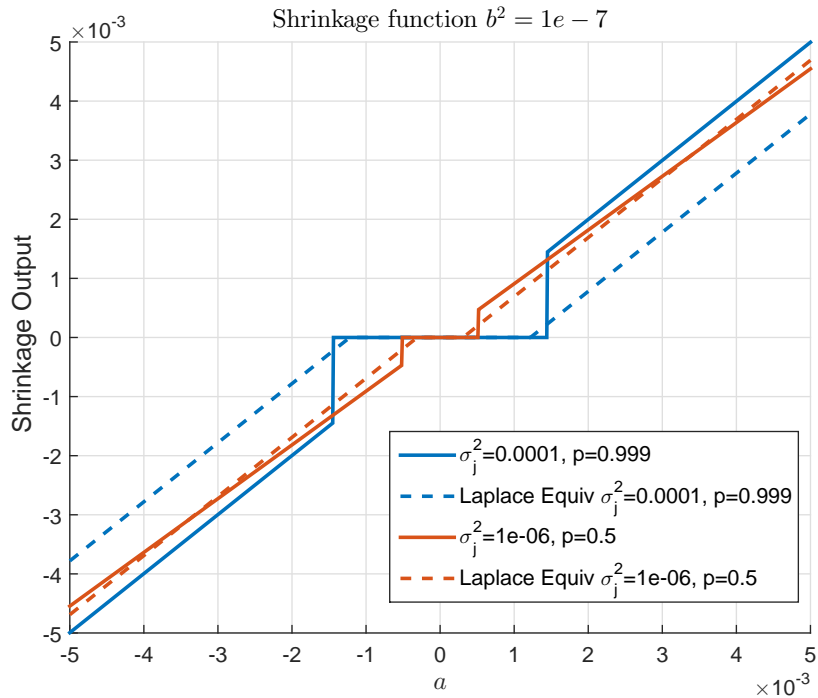}
    \caption{Shrinkage Functions of the spike and slab and the corresponding $b^{2}-equivalent$ Laplace prior}
    \label{fig:SpikeSlabShrinkCompare}
\end{figure}
\begin{remark}
The use of Laplace priors and $\ell_{1}$ penalties has been applied in context of robust Kalman filtering and smoothing in \cite{BoydRobustKalman,AravkinL1Laplace}.   Here the authors considered the problem of non-gaussian heavy tailed observation noise rather than process noise.
\end{remark}

\section{Bayesian Approach using MCMC}\label{secMCMC}
In this section we consider a fully Bayesian approach where we estimate the posterior distribution of $\Gamma$.   The advantages of the fully Bayesian approach to this problem are
\begin{enumerate}
  \item Uncertainty in nuisance parameters such as $J$ and $\sigma_{o}^{2}$ are averaged out rather than relying on MAP point estimates
  \item Estimate of the posterior distribution of $\Gamma$ is obtained which provides more information than a posterior mode
  \item Estimates of uncertainty in covariance estimate can be obtained.
\end{enumerate}

To describe the Bayesian approach to estimation of $\Gamma$, let $\Phi$ represent the unknown parameters
\begin{eqnarray}
\Phi_{0} &=& Y_{miss} \nonumber \\
\Phi_{1} &=& X(1:T) \nonumber \\
\Phi_{2} &=& J(2:T) \nonumber \\
\Phi_{3} &=& D \nonumber \\
\Phi_{4} &=& \sigSet \nonumber \\
\Phi_{5} &=& \zeta \nonumber \\
\Phi_{6} &=& \left\{\sigma_{j,i}^{2}(t)\right\}_{i=1,\dots,N,t=1,\dots,T}
\end{eqnarray}
where $Y_{miss}$ are the unobserved prices.   Unlike the KECM approaches in the previous section we sample the missing observations $Y_{miss}$.   One advantage of sampling $Y_{miss}$ is that the covariance of
$X(t)|X(t-1),X(t+1),Y(t),Y_{miss}(t)$ is the same for all $2 \le t \le T-1$, where as the covariance of $X(t)|X(t-1),X(t+1),Y(t)$ depends on $t$.   This simplification allows for faster numerical simulation in the Gibbs sampler.

In the Bayesian approach given the data $Y(t)$ we wish to compute the posterior distribution of $\Gamma$
\begin{equation}\label{eqPosteriorMCMC}
   p(\gamma|y) = \frac{ \int p(y|\phi,\gamma)p(\phi,\gamma) d\phi}{\int\int p(y|\phi,\gamma')p(\phi,\gamma')d\phi d\gamma'}.
\end{equation}
Once the posterior is obtained the posterior mean of $\Gamma$ can be obtained via
\begin{equation}\label{eqPosteriorMeanGamma}
   \E_{|y}\Gamma = \int \gamma p(\gamma|y) d\gamma.
\end{equation}
The posterior mean which is optimal in a minimum mean squared error (MMSE) sense can be used as an estimate of $\Gamma$.

Evaluating the integrals in \eqref{eqPosteriorMCMC} and \eqref{eqPosteriorMeanGamma} are intractable, however we can obtain samples from the posterior distribution using a Markov Chain Monte Carlo (MCMC) technique such as Gibbs sampling.   These samples can then be used to obtain an estimate of $\E_{|y}\Gamma$.  A Gibbs sampling approach for estimating $\E_{|y}\Gamma$ is described in the next section.

\subsection{Gibbs Sampling approach}\label{SecGibbs}
Gibbs sampling \cite{ExplainingGibbs} is an MCMC approach for generating samples from a multivariable distributions such as $p(\phi,\gamma|y)$.
In this application Gibbs sampling may be implemented as follows to generate samples of $\Gamma^{(1)}, \dots, \Gamma^{(M_{G})}$ from $p(\gamma|y)$.
\begin{enumerate}
  \item Initialize the first samples,$\Phi^{(0)}$,$\Gamma^{(0)}$
  \item for $k = 1$ to $M_{G}$
  \begin{itemize}
    \item Sample $Y_{miss}^{(k)}$ from the conditional distribution $p(Y_{miss}|X^{(k-1)},\sigma_{o}^{2})$
    \item for $t=1$ to $T$
    \begin{itemize}
        \item Sample $X(t)^{(k)}$ from
        \begin{equation*}
          p(x(t)|y,Y_{miss}^{(k)},\Phi^{(k-1)}_{2:6},X^{(k)}(1:t-1),X^{(k-1)}(t+1:T),\Gamma^{(k-1)})
        \end{equation*}

    \end{itemize}
    \item for $l=1$ to $L$
    \begin{itemize}
        \item for $t=1$ to $T$,$n=1$ to $N$
        \begin{itemize}
        \item Sample $J_{n}(t)^{(k-1+l/L)}$ from
        \end{itemize}
         \begin{equation*}
           p(j_{n}(t)|y,\Phi^{(k)}_{0,1},J_{1:n-1}^{(k-1+l/L)}(t),J_{n+1:N}^{(k-1+ (l-1)/L)}(t),\Gamma^{(k-1)})
         \end{equation*}
    \end{itemize}
    \item Sample $D$ from  $p(d|y,\Phi^{(k)}_{0:2},\Phi^{(k-1)}_{4:6},\Gamma^{(k-1)})$
    \item Sample $\Gamma^{(k)}$ from $p(\gamma|y,\Phi^{(k)}_{0:3},\Phi^{(k-1)}_{4:6})$
    \item Sample $\sigma_{o}^{2,(k)}$ from  $p(\sigma_{o}^{2}|y,\Phi^{(k)}_{0:3},\Phi^{(k-1)}_{5:6},\Gamma^{(k)})$
    \item Sample $\zeta^{(k)}$ from  $p(\zeta|y,\Phi^{(k)}_{0:4},\Phi^{(k-1)}_{6},\Gamma^{(k)})$
    \item Sample $\sigma_{j,i}^{2,(k)}(t)$ from  $p(\sigma_{j,i}^{2}(t)|y,\Phi^{(k)}_{0:5},\Gamma^{(k)})$  for all $i,t$
  \end{itemize}
\end{enumerate}
where $\Phi_{i:j}$ refers to $[\Phi_{i},\dots, \Phi_{j}]$ and where
\begin{equation*}
  \Phi_{-n}(t)=[\Phi_{0}(t),\dots, \Phi_{n-1}(t),\Phi_{n+1}(t),\dots,\Phi_{6}(t)]
\end{equation*}
Each of these steps draws from conditional distributions can be implemented easily as shown in Appendix \ref{SecMCMCAppendix}.

It can be shown using well known results on Markov chains that the samples produced by the above Gibbs sampler form a Markov chain \cite{MonteCarloCasella} with a limiting stationary distribution $p(\phi,\gamma|y)$.
\subsection{Estimation of $\Gamma$}
The samples of $\Gamma^{(k)}$ are used as an estimate of the posterior distribution of $\Gamma$.   Using the estimated posterior distribution the posterior mean of $\Gamma$ is the sample average of $\Gamma^{(k)}$ ( where we discarded earlier samples to allow for the samples to converge )
\begin{equation*}
  \hat{\Gamma} = \frac{1}{M_{G}-k+1}\sum_{m=k}^{M_{G}} \Gamma^{(m)}.
\end{equation*}

Another technique to estimate $\Gamma$ is Rao-Blackwellization which reduces the variance in the covariance estimate \cite{CovarLiu}.
Here we take the sample of average of the conditional means to arrive at an estimate of the posterior mean of $\Gamma$
\begin{equation*}
  \hat{\Gamma} = \frac{1}{M_{G}-k+1}\sum_{m=k}^{M_{G}} \E(\Gamma|\Phi^{(k)}).
\end{equation*}
The numerical experiments presented in the next section use Rao-Blackwellization for the posterior mean estimation.

\section{Numerical Results}\label{sec:ECM_numerical}
\definecolor{LightCyan}{rgb}{1,0.97,0.86}
In this section we evaluate the performance of the following algorithms
 \begin{enumerate}
   \item KEM \cite{KalmanEMStocks}
   \item KECM Laplace(section \ref{Sec:ECM})
   \item KECM Spike and Slab (section \ref{Sec:ECM})
   \item MCMC approach (section \ref{secMCMC})
   \item Pairwise refresh with TSCV \cite{VastVolMatrix,ZhangEpps}
   \item Pairwise refresh with TSCV and jump correction \cite{BoudtJump}
 \end{enumerate}
for determining a covariance matrix from high frequency data.
The performance is evaluated using a Monte Carlo approach with simulated high frequency return data.
\subsection{Performance Assessment Methodology}
We track two performance measures for the covariance estimate, $\hat{\Gamma}$, in this study.  For the first performance measure we compute the minimum variance portfolio
\begin{align*}
\tilde{w} =\arg\min_{w} w^{T}\hat{\Gamma} w \\
\textnormal{s.t.} \sum_{i} w_{i} =1.
\end{align*}
The variance of this portfolio's return is then computed as a figure of merit.  The variance of the portfolio return is given below
\begin{equation*}
  \tilde{w}^{T}\Gamma \tilde{w}.
\end{equation*}

For the second performance measure we compute the relative Frobenius norm of the error between the true and estimated covariance
\begin{equation*}
  \frac{ \sqrt{\sum_{i,j}|\Gamma_{i,j}-\hat{\Gamma}_{i,j}|^{2}}}{\sqrt{\sum_{i,j}|\Gamma_{i,j}|^{2}}}.
\end{equation*}

\subsection{Algorithm Initialization and other considerations}
In each study we initialize the algorithms in the same way.   The hyper-parameters for the prior distribution are listed in Table \ref{TableNumSim_Hyper}.   For the KEM and KECM algorithms the initial covariance estimate is computed using the time refresh method in \cite{RefreshTime}.   The initialization of drift and jump estimate of each algorithm is set to zero.   For the MCMC algorithm we take the output of the KECM spike and slab algorithm as the first sample.

In the KECM algorithms we employ one additional initialization step to avoid being trapped in an over-smoothed local solution.  This step involves using a forward Kalman filter rather than a smoother to approximate the posterior distribution of $X(t)$ in the first 10 iteration of the KECM algorithms.   After 10 iterations we revert to the approaches described in Section \ref{Sec:ECM} which use the Kalman smoother.

The stability of the covariance estimate forms the basis for a stopping criterion in the KECM algorithms.  The KECM algorithms are terminated at iteration $n$ when the relative difference between the current and previous covariance estimate is less than 0.001
\begin{equation*}
  \frac{ \sqrt{\sum_{i,j}|\hat{\Gamma}_{i,j}^{(n)}-\hat{\Gamma}_{i,j}^{(n-1)}|^{2}}}{\sqrt{\sum_{i,j}|\hat{\Gamma}_{i,j}^{(n-1)}|^{2}}} < 0.001.
\end{equation*}
For the MCMC algorithm we generate 10000 samples and discard the first 2000 samples to allow for convergence of the Markov chain.

Since jumps cannot be predicted an ambiguity occurs if there is no observation of the price at the time the jump occurs.   Thus to prevent ambiguity we assume jumps in the $i^{th}$ asset price can only occur if an observation of the $i^{th}$ price is made.   We believe that this is a mild assumption given that in many markets jumps in the efficient price will be traded upon almost immediately.  This assumption is built into the KECM and MCMC approaches by setting $\lambda=\infty$ and $\zeta=1$ when an observation does not occur.

\subsection{Simulated Data Jump Model}
For the data study we simulated 30 minutes of data from 20 assets according to equations \eqref{eqX} and \eqref{eqPriceObserved} at 1 second intervals.  Here 50 data sets were generated to test our algorithms.
Taking motivation from factor models for U.S. stock returns we set our covariance $\Gamma$ according to the following 5 factor model
\begin{equation*}
\Gamma = \sum_{i=1}^{5} \beta_{v_{i}}v_{i}v_{i}^{T}  + \epsilon I.
\end{equation*}
Here we compute a new covariance for each Monte Carlo data set.
We draw $v_{1}$ from a multivariate normal distribution with mean $\frac{1}{\sqrt{2}}$ and covariance $0.5I$.   For $i>1$, we draw $v_{i}$ from a multivariate normal distribution with zero mean and covariance $I$.  The factor variance $\beta_{v_{i}}$ is modeled as gamma distributed with shape 2 and mean $\frac{0.7 * 0.02^{2}}{23400}$ for $i=1$ and mean $\frac{0.3/4 * 0.02^{2}}{23400}$ for $i \ne 1$.  The $\epsilon$ term is defined to be $\frac{0.02^{2}}{23400*100}$.   With these settings each simulated asset will on average have a daily return volatility of approximately 2 percent.

For the $D$ parameter we use a random number generator for each data set.   The value for $D$ was drawn from a multi-variate normal distribution with mean 0 and covariance $\left(\frac{0.01}{23400}\right)^{2}I$.   The observation noise variance of each asset was set to a random number drawn from a gamma distribution with shape 2 and mean $0.0002^{2}$.  For a stock price of $\$25$ this corresponds to a mean noise standard deviation of about $\$0.005$.  The jump parameters $\zeta$ and $\sigma_{j}^{2}$ were varied parametrical over several values.

Both the KECM and MCMC algorithms require hyperparameters to be specified for the prior distributions.  For these experiments we choose hyperparameters which would result in diffuse priors in order to minimize bias.  For the hyperparameters of the Laplace prior in the KECM algorithm we used the technique described in Section \ref{Sec:ProcedureQlam}.   A listing of all the hyperparameters used in the algorithms are shown in Table \ref{TableNumSim_Hyper}.

\definecolor{GrayTable}{gray}{0.85}

\newcolumntype{a}{>{\columncolor{GrayTable}}c}
\newcolumntype{b}{>{\columncolor{white}}c}

\begin{table}
\centering
\begin{tabular}{|l | b | b | b |}
\hline

  & {Value} & {Comment} \\
\hline
$\alpha_{\zeta}$ & $10 \times 0.995$ &    \\ \hline
$\beta_{\zeta}$  &$10 - \alpha_{\zeta}$ &  prior mean of $\zeta$ is $0.995$  \\ \hline
$\alpha_{j}$ & 10 & \\ \hline
$\beta_{j}$ &$0.01^{2}(\alpha_{j}+1)$ &  prior mode of $\sigma_{j}^{2}$ is 1e-4 \\ \hline
$\alpha_{o}$ &5 & \\ \hline
$\beta_{o}$ & $(\alpha_{o}+1) \times 0.0001^{2}$ & prior mode of $\sigma_{o}^{2}$ is 1e-8  \\ \hline
$\eta$ &$N+5$ &  \\ \hline
$W_{o}$ &$\frac{0.02^{2}(\eta+N+1)}{23400}I$ & Corresponds to 0.02\% daily volatility \\ \hline
$\alpha_{\lambda}$ &5.6 & Obtained using method in Section \ref{Sec:Laplace} \\ \hline
$\beta_{\lambda}$ &5e-04 &Obtained using method in Section \ref{Sec:Laplace} \\ \hline
\end{tabular}
\caption{Parameters used in KEM, KECM and MCMC algorithms}
\label{TableNumSim_Hyper}
\end{table} 
The probability that any given price is observed is set to be commensurate with the price innovation.   This is consistent with empirical observations that trading volume can be positively correlated with volatility \cite{Karpoff}.   To model this association the probability that the $m^{th}$ asset price will be observed at time $t$ is simulated as
\begin{equation*}
  p_{obs,m}(t)=\frac{|X_{m}(t)-X_{m-1}(t)-D_{m}|}{|X_{m}(t)-X_{m-1}(t)-D_{m}|+\nu}
\end{equation*}
where
\begin{equation*}
  \nu = \frac{\sqrt{2\Gamma_{m,m}}}{\pi}\left(\frac{1}{p_{Obs}}-1\right).
\end{equation*}
This choice of $\nu$ ensures that when the innovation achieves its mean absolute value ,$\sqrt{\frac{2\Gamma_{m,m}}{\pi}}$, the probability of an observation will be $p_{Obs}$. We set $p_{Obs}=0.3$ in our numerical experiments.

The performance results for different values of the jump parameters are shown Tables \ref{TableVar_J} and \ref{TableCovar_J}.   For the majority of cases we see that the KECM approaches outperform the other methods when jumps are present.   In Figure \ref{fig:KEMjumpMiss} we show the Kalman estimate of the true price for various algorithms. The figure highlights the disadvantage of the KEM algorithm in the presence of jumps, namely that it over smoothes prices near jumps.  

\definecolor{GrayTable}{gray}{0.85}

\definecolor{Pass}{rgb}{0,1,0}
\definecolor{Fail}{rgb}{1,0,0}

\newcolumntype{a}{>{\columncolor{GrayTable}}c}
\newcolumntype{b}{>{\columncolor{white}}c}

\begin{table}
\centering
\begin{tabular}{|l |l |b | b | b | b|b|b|}
\hline
&  &  &KECM  &KECM  & & Pairwise  & Pairwise  \\
$\zeta $      & $\sigma_j^{2}$  & KEM     & Laplace & Spike     & MCMC &Refresh & Refresh  \\
& & & & \& Slab & & &(jump) \\
\hline
1 & N/A&\cellcolor{Pass}1.2e-10 & 1.3e-10 & 1.3e-10 & 1.3e-10 & 1.8e-10 & 2e-10 \\ \hline
0.9999 & 6.25e-06 & 1.5e-10 & 1.4e-10&\cellcolor{Pass}1.4e-10 & 1.5e-10 & 1.8e-10 & 2e-10 \\ \hline
0.9999 & 0.0001 & 1.6e-10 & 1.4e-10&\cellcolor{Pass}1.4e-10 & 1.5e-10 & 2.6e-10 & 2.2e-10 \\ \hline
0.9995 & 6.25e-06 & 1.6e-10 & 1.3e-10&\cellcolor{Pass}1.3e-10 & 1.3e-10 & 2.4e-10 & 2.3e-10 \\ \hline
0.9995 & 0.0001 & 3e-10 & 1.3e-10&\cellcolor{Pass}1.2e-10 & 1.3e-10 & 7.9e-10 & 6e-10 \\ \hline
0.999 & 6.25e-06 & 2.4e-10 & 1.6e-10&\cellcolor{Pass}1.6e-10 & 1.7e-10 & 4.7e-10 & 4.5e-10 \\ \hline
0.999 & 2.5e-05 & 4.5e-10 & 1.7e-10&\cellcolor{Pass}1.7e-10 & 1.8e-10 & 9.8e-10 & 6.9e-10 \\ \hline
0.999 & 0.0001 & 8.2e-10&\cellcolor{Pass}1.6e-10 & 1.6e-10 & 1.7e-10 & 1.7e-09 & 1.1e-09 \\ \hline
\end{tabular}
\caption{Portfolio variance for jump model, best performance highlighted in green.}
\label{TableVar_J}
\end{table}

\begin{table}
\centering
\begin{tabular}{|l |l |b | b | b | b|b|b|}
\hline
&  &  &KECM   &KECM & & Pairwise  & Pairwise  \\
$\zeta $      & $\sigma_j^{2}$  & KEM      & Laplace & Spike     & MCMC &Refresh  & Refresh  \\
& & & & \& Slab & & &(jump) \\
\hline
1 & N/A&\cellcolor{Pass}0.2 & 0.2 & 0.2 & 0.22 & 0.48 & 0.49 \\ \hline
0.9999 & 6.25e-06 & 0.22&\cellcolor{Pass}0.22 & 0.22 & 0.24 & 0.47 & 0.49 \\ \hline
0.9999 & 0.0001 & 0.73&\cellcolor{Pass}0.21 & 0.21 & 0.22 & 0.89 & 0.69 \\ \hline
0.9995 & 6.25e-06 & 0.29&\cellcolor{Pass}0.21 & 0.21 & 0.22 & 0.55 & 0.54 \\ \hline
0.9995 & 0.0001 & \cellcolor{Fail}3.5 & 0.18&\cellcolor{Pass}0.18 & 0.2 & \cellcolor{Fail}2.9 & \cellcolor{Fail}2.1 \\ \hline
0.999 & 6.25e-06 & 0.36 & 0.21&\cellcolor{Pass}0.21 & 0.22 & 0.67 & 0.65 \\ \hline
0.999 & 2.5e-05 & \cellcolor{Fail}1.1&\cellcolor{Pass}0.21 & 0.21 & 0.22 & \cellcolor{Fail}1.5 & \cellcolor{Fail}1.4 \\ \hline
0.999 & 0.0001 & \cellcolor{Fail}4.8 & 0.2&\cellcolor{Pass}0.2 & 0.21 & \cellcolor{Fail}4.6 & \cellcolor{Fail}3.6 \\ \hline
\end{tabular}
\caption{Average covariance error for jump model, best performance highlighted in green.  Large errors highlighted in red.}
\label{TableCovar_J}
\end{table}

\begin{figure}[h]
    \centering
    \includegraphics[width=4.7in]{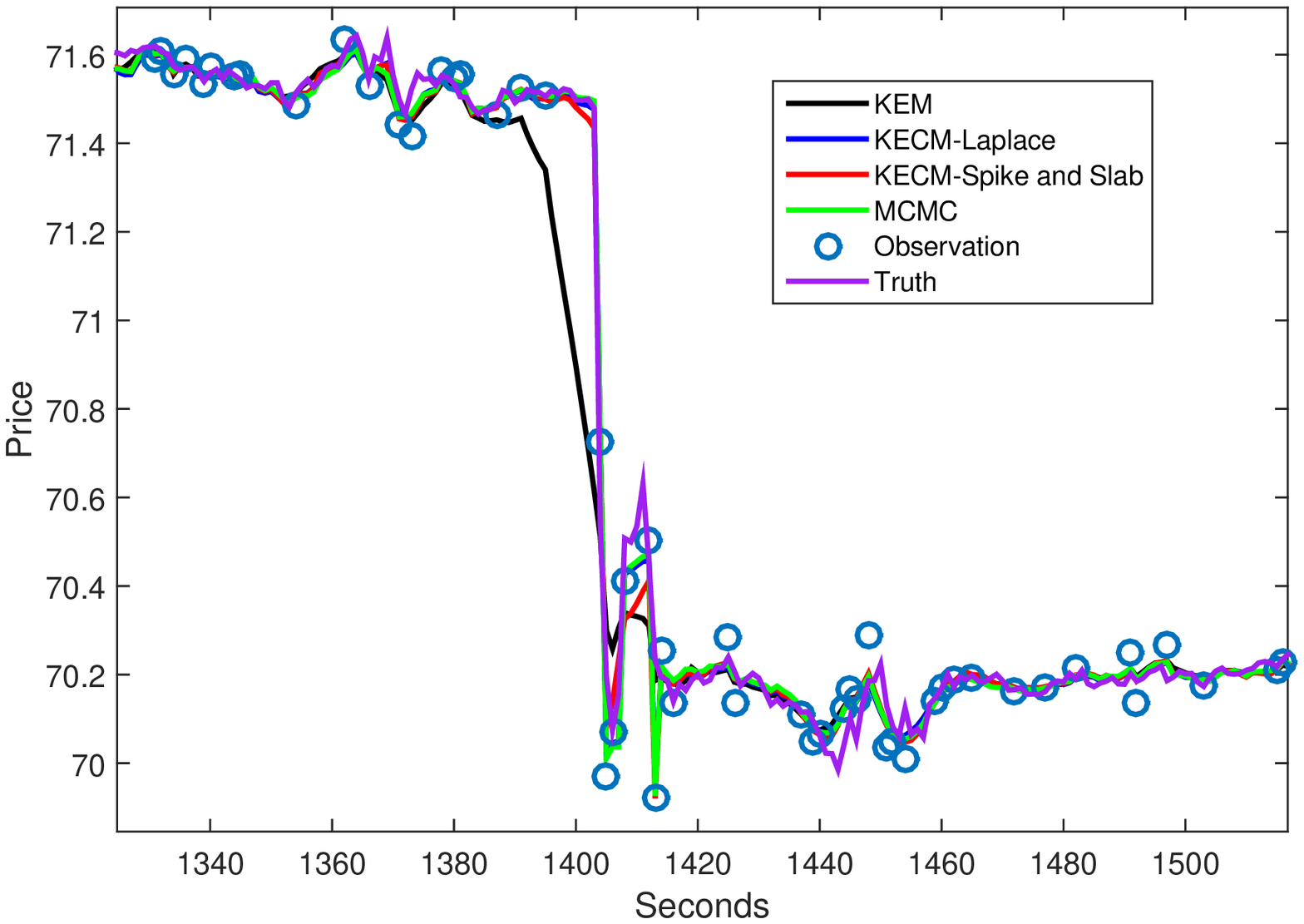}
    \caption{Price estimate example from the KEM, KECM, and Gibbs sampling.   This is an example of the KEM algorithm over-smoothing near a small jump in price}
    \label{fig:KEMjumpMiss}
\end{figure}
\subsection{Simulated Data from GARCH(1,1)-jump model}\label{Sec:GarchJump}
In addition to the jump diffusion model we also evaluate the algorithms against a multivariate GARCH(1,1)-jump pricing model \cite{ChanGARCHJump,MaheuGARCHJump,Bollerslev1990}, where the effect of jumps persists in the price volatility.   Using the GARCH(1,1)-jump model the log-price price data is generated as
\begin{equation*}
  X_{i}(t)=X_{i}(t-1) + \sqrt{h_{i}}V_{i}(t) +J_{i}(t)Z_{i}(t) +D
\end{equation*}
\begin{equation*}
  h_{i}(t+1)=b_{i}h_{i}(t)+a_{i}(X_{i}(t)-X_{i}(t-1)-D)^{2} + c_{i}
\end{equation*}
\begin{equation*}
  h_{i}(0)=\Gamma_{i,i}
\end{equation*}
where $a_{i},b_{i},c_{i}$ are non-negative with $b_{i}+a_{i} <1$ and $c_{i}= \Gamma_{i,i}(1-a_{i}-b_{i})$.  Here $V(t)$ is modeled as multivariate normal with
\begin{itemize}
  \item $V_{i}(t) \sim \mathcal{N}(0,1)$
  \item $\E V_{i}(t)V_{j}(t)=\frac{\Gamma_{i,j}}{\sqrt{\Gamma_{i,i}\Gamma_{j,j}}}$
  \item $\E V_{i}(t_{1})V_{j}(t_{2})=0$ for $t_{1} \ne t_{2}$.
\end{itemize}
The value of $c$ ensures that in the absence of jumps, the long term average volatility for the $i^{th}$ asset will be $\sqrt{\Gamma_{i,i}}$.  We also see that the correlation coefficient between any two assets is constant \cite{Bollerslev1990}.

In these experiments $a_{i}=0.3$ and $b_{i}=0.5$.   This allows for volatility clustering which has been observed in many empirical stock return data.    All other parameters such as the covariance matrix are identical to the previous experiment.

The results for the GARCH(1,1)-jump model are shown in Tables \ref{TableVar_GJ} - \ref{TableCovar_GJ}.   From these tables we see that the KECM and MCMC algorithms are robust to the volatility clustering exhibited in GARCH models.   

\definecolor{GrayTable}{gray}{0.85}

\definecolor{Pass}{rgb}{0,1,0}
\definecolor{Fail}{rgb}{1,0,0}

\newcolumntype{a}{>{\columncolor{GrayTable}}c}
\newcolumntype{b}{>{\columncolor{white}}c}

\begin{table}
\centering
\begin{tabular}{|l |l |b | b | b | b|b|b|}
\hline
&  &  &KECM  &KECM  & & Pairwise  & Pairwise  \\
$\zeta $      & $\sigma_j^{2}$  & KEM     & Laplace & Spike     & MCMC &Refresh & Refresh  \\
& & & & \& Slab & & &(jump) \\
\hline
1 & N/A&\cellcolor{Pass}1.3e-10 & 1.3e-10 & 1.3e-10 & 1.4e-10 & 2.5e-10 & 2.6e-10 \\ \hline
0.9999 & 6.25e-06 & 1.6e-10 & 1.6e-10&\cellcolor{Pass}1.5e-10 & 1.6e-10 & 2.4e-10 & 2.5e-10 \\ \hline
0.9999 & 0.0001 & 1.6e-10&\cellcolor{Pass}1.3e-10 & 1.3e-10 & 1.3e-10 & 4.4e-10 & 3.1e-10 \\ \hline
0.9995 & 6.25e-06 & 2e-10&\cellcolor{Pass}1.5e-10 & 1.5e-10 & 1.5e-10 & 4.4e-10 & 3.9e-10 \\ \hline
0.9995 & 0.0001 & 3.7e-10&\cellcolor{Pass}1.3e-10 & 1.4e-10 & 1.3e-10 & 1e-09 & 4.5e-10 \\ \hline
0.999 & 6.25e-06 & 2.6e-10&\cellcolor{Pass}1.4e-10 & 1.4e-10 & 1.4e-10 & 5.8e-10 & 4.7e-10 \\ \hline
0.999 & 2.5e-05 & 5.5e-10&\cellcolor{Pass}1.5e-10 & 1.7e-10 & 1.6e-10 & 1.4e-09 & 8.1e-10 \\ \hline
0.999 & 0.0001 & 1.1e-09 & 1.6e-10 & 1.5e-10&\cellcolor{Pass}1.5e-10 & 2e-09 & 1e-09 \\ \hline
\end{tabular}
\caption{Portfolio variance for GARCH(1,1)-jump model, best performance highlighted in green.}
\label{TableVar_GJ}
\end{table}

\begin{table}
\centering
\begin{tabular}{|l |l |b | b | b | b|b|b|}
\hline
&  &  &KECM   &KECM & & Pairwise  & Pairwise  \\
$\zeta $      & $\sigma_j^{2}$  & KEM      & Laplace & Spike     & MCMC &Refresh  & Refresh  \\
& & & & \& Slab & & &(jump) \\
\hline
1 & N/A&\cellcolor{Pass}0.37 & 0.37 & 0.38 & 0.41 & 0.5 & 0.51 \\ \hline
0.9999 & 6.25e-06 & 0.43&\cellcolor{Pass}0.37 & 0.38 & 0.41 & 0.58 & 0.55 \\ \hline
0.9999 & 0.0001 & \cellcolor{Fail}3.3&\cellcolor{Pass}0.39 & 0.4 & 0.43 & \cellcolor{Fail}1.7 & 0.55 \\ \hline
0.9995 & 6.25e-06 & 0.88&\cellcolor{Pass}0.42 & 0.43 & 0.44 & 0.81 & 0.62 \\ \hline
0.9995 & 0.0001 & \cellcolor{Fail}18 & 0.65 & 0.49&\cellcolor{Pass}0.46 & \cellcolor{Fail}8.5 &\cellcolor{Fail} 2.2 \\ \hline
0.999 & 6.25e-06 & \cellcolor{Fail}1.4&\cellcolor{Pass}0.48 & 0.51 & 0.49 & \cellcolor{Fail}1.2 & 0.85 \\ \hline
0.999 & 2.5e-05 & \cellcolor{Fail}7.7 & 0.64 & 0.62&\cellcolor{Pass}0.52 & \cellcolor{Fail}4.5 & \cellcolor{Fail}1.8 \\ \hline
0.999 & 0.0001 & \cellcolor{Fail}36 & \cellcolor{Fail}1.4 & 0.67&\cellcolor{Pass}0.55 & \cellcolor{Fail}16 & \cellcolor{Fail}6.1 \\ \hline
\end{tabular}
\caption{Average covariance error for GARCH(1,1)-jump model, best performance highlighted in green.   Large errors highlighted in red.}
\label{TableCovar_GJ}
\end{table}

\subsection{Simulated Data from GARCH(1,1)-jump Model and stochastic microstructure variance}
In this section we test our algorithms under a GARCH(1,1)-jump model with stochastic microstructure variance.  This microstructure noise model accounts for a positive correlation between the bid-ask spread and the squared innovation.   This models an empirical phenomena that has been observed in many markets \cite{BidAskZhang}.   Here we assume the same efficient price innovation as the GARCH(1,1)-jump model but now we allow for time-varying variance in the microstructure noise.   In this model the variance of the microstructure noise at time $t$ for $i^{th}$ asset is
\begin{equation*}
  \left(0.1\frac{(X_{i}(t)-X_{i}(t-1)-D)^{2}}{\Gamma_{i,i}} + 0.9 \right)\tilde{\sigma}_{o,i}^{2}
\end{equation*}
which is the sum of fixed variance and time varying term which is dependent on the efficient price innovation.
Here we see that when the squared innovation equals the variance then the observation noise variance equals $\tilde{\sigma}_{o,i}^{2}$.   As in the previous simulations, $\tilde{\sigma}_{o,i}^{2}$ is chosen to be a realization of a gamma distributed random variable with shape 2 and mean $0.0002^{2}$.

The results for this model are shown in Tables \ref{TableVar_GJMN} and \ref{TableCovar_GJMN}.   A comparison of the covariance errors is shown in Table \ref{TableCovarCompare}.   From the comparison table we see that the covariance errors are larger for the non-stationary microstructure noise model.   Here the KECM-Laplace model is especially sensitive to the stochastic microstructure noise variance for $\sigma_{j}^{2}=1e-4$.   In some cases the covariance error increased by about a factor of 10.   The KECM-spike and slab and MCMC approaches were not as sensitive to the stochastic noise variance.

\definecolor{GrayTable}{gray}{0.85}

\definecolor{Pass}{rgb}{0,1,0}
\definecolor{Fail}{rgb}{1,0,0}

\newcolumntype{a}{>{\columncolor{GrayTable}}c}
\newcolumntype{b}{>{\columncolor{white}}c}

\begin{table}
\centering
\begin{tabular}{|l |l |b | b | b | b|b|b|}
\hline
&  &  &KECM  &KECM  & & Pairwise  & Pairwise  \\
$\zeta $      & $\sigma_j^{2}$  & KEM     & Laplace & Spike     & MCMC &Refresh & Refresh \\
& & & & \& Slab & & & (jump) \\
\hline
1 & N/A&\cellcolor{Pass}1.5e-10 & 1.6e-10 & 1.6e-10 & 1.6e-10 & 2.2e-10 & 2.5e-10 \\ \hline
0.9999 & 6.25e-06&\cellcolor{Pass}1.6e-10 & 1.6e-10 & 1.6e-10 & 1.6e-10 & 3e-10 & 3e-10 \\ \hline
0.9999 & 0.0001 & 2e-10&\cellcolor{Pass}1.6e-10 & 1.6e-10 & 1.6e-10 & 5.1e-10 & 3.7e-10 \\ \hline
0.9995 & 6.25e-06 & 2.6e-10&\cellcolor{Pass}1.9e-10 & 1.9e-10 & 1.9e-10 & 4.6e-10 & 4.2e-10 \\ \hline
0.9995 & 0.0001 & 5.1e-10 & 1.8e-10 & 1.8e-10&\cellcolor{Pass}1.7e-10 & 1.5e-09 & 6.2e-10 \\ \hline
0.999 & 6.25e-06 & 2.3e-10&\cellcolor{Pass}1.5e-10 & 1.5e-10 & 1.5e-10 & 5e-10 & 4.2e-10 \\ \hline
0.999 & 2.5e-05 & 5.6e-10 & 1.7e-10 & 1.7e-10&\cellcolor{Pass}1.7e-10 & 1.4e-09 & 9.1e-10 \\ \hline
0.999 & 0.0001 & 9e-10 & 2e-10 & 1.6e-10&\cellcolor{Pass}1.5e-10 & 2.3e-09 & 9.3e-10 \\ \hline
\end{tabular}
\caption{Portfolio variance for GARCH(1,1)-jump model with stochastic microstructure noise variance, best performance highlighted in green}
\label{TableVar_GJMN}
\end{table}

\begin{table}
\centering
\begin{tabular}{|l |l |b | b | b | b|b|b|}
\hline
&  &  &KECM   &KECM & & Pairwise  & Pairwise  \\
$\zeta $      & $\sigma_j^{2}$  & KEM      & Laplace & Spike     & MCMC &Refresh  & Refresh \\
& & & & \& Slab & & & (jump)\\
\hline
1 & N/A & 0.37&\cellcolor{Pass}0.37 & 0.38 & 0.42 & 0.57 & 0.58 \\ \hline
0.9999 & 6.25e-06 & 0.51&\cellcolor{Pass}0.42 & 0.42 & 0.47 & 0.55 & 0.53 \\ \hline
0.9999 & 0.0001 & \cellcolor{Fail}21 & \cellcolor{Fail}1.5&\cellcolor{Pass}0.38 & 0.41 & \cellcolor{Fail}2.6 & 0.65 \\ \hline
0.9995 & 6.25e-06 & 0.78&\cellcolor{Pass}0.4 & 0.41 & 0.44 & 0.82 & 0.67 \\ \hline
0.9995 & 0.0001 & \cellcolor{Fail}75 & \cellcolor{Fail}3.3 & 0.47&\cellcolor{Pass}0.44 & \cellcolor{Fail}9.5 & \cellcolor{Fail}2.6 \\ \hline
0.999 & 6.25e-06 & 1.2&\cellcolor{Pass}0.41 & 0.44 & 0.43 & \cellcolor{Fail}1.0 & 0.71 \\ \hline
0.999 & 2.5e-05 & \cellcolor{Fail}13 & 0.48 & 0.51&\cellcolor{Pass}0.46 & \cellcolor{Fail}3.5 & \cellcolor{Fail}1.6 \\ \hline
0.999 & 0.0001 & \cellcolor{Fail}130 & \cellcolor{Fail}13 & \cellcolor{Fail}2.7&\cellcolor{Pass}0.45 & \cellcolor{Fail}13 & \cellcolor{Fail}4.7 \\ \hline
\end{tabular}
\caption{Average covariance error for GARCH(1,1)-jump model with stochastic microstructure noise variance, best performance highlighted in green.   Large errors highlighted in red.}
\label{TableCovar_GJMN}
\end{table}

\definecolor{GrayTable}{gray}{0.85}

\definecolor{Pass}{rgb}{0,1,0}
\definecolor{Fail}{rgb}{1,0,0}

\newcolumntype{a}{>{\columncolor{GrayTable}}c}
\newcolumntype{b}{>{\columncolor{white}}c}

\begin{table}
\centering
\begin{tabular}{|l |l |b | b | b | b|b|b|}
\hline
&  & KECM  &KECM  &KECM & KECM & MCMC  & MCMC  \\

$\zeta $      & $\sigma_j^{2}$   & Laplace & Laplace & Spike    & Spike   & &  \\
&  &   &   & \&Slab  & \&Slab & &   \\
&     &Model 1   &Model 2 & Model 1   &Model 2 &Model 1   &Model 2 \\
\hline
1 & N/A & 0.37 & 0.37 & 0.38 & 0.38 & 0.41 & 0.42 \\ \hline
0.9999 & 6.25e-06 & 0.37 & 1.5 & 0.38 & 0.38 & 0.41 & 0.41 \\ \hline
0.9999 & 0.0001 & 0.39 & 0.42 & 0.4 & 0.42 & 0.43 & 0.47 \\ \hline
0.9995 & 6.25e-06 & 0.42 & 3.3 & 0.43 & 0.47 & 0.44 & 0.44 \\ \hline
0.9995 & 0.0001 & 0.65 & 0.4 & 0.49 & 0.41 & 0.46 & 0.44 \\ \hline
0.999 & 6.25e-06 & 0.48 & 13 & 0.51 & 2.7 & 0.49 & 0.45 \\ \hline
0.999 & 2.5e-05 & 0.64 & 0.48 & 0.62 & 0.51 & 0.52 & 0.46 \\ \hline
0.999 & 0.0001 & 1.4 & 0.41 & 0.67 & 0.44 & 0.55 & 0.43 \\ \hline
\end{tabular}
\caption{Average covariance error comparison for GARCH-jump model (model 1) and GARCH-jump model with with stochastic microstructure noise variance (model 2)}
\label{TableCovarCompare}
\end{table}

\subsection{Timing}\label{sec:TimingECM}
Average MATLAB timing of the algorithms for the GARCH(1,1)-jump model with stochastic microstructure noise variance are shown in Table \ref{TableTiming}.   The machine running the simulation has the Windows 7 operating system and an Intel i7-3740 processor with 32.0 GB of RAM.   The table shows that the pairwise refresh methods are the least computationally costly, while the MCMC method requires the most run time.   The data also indicates that the KEM, KECM, and MCMC algorithms take longer to converge when larger and more frequent jumps are present.

\definecolor{GrayTable}{gray}{0.85}

\definecolor{Pass}{rgb}{0,1,0}
\definecolor{Fail}{rgb}{1,0,0}

\newcolumntype{a}{>{\columncolor{GrayTable}}c}
\newcolumntype{b}{>{\columncolor{white}}c}

\begin{table}
\centering
\begin{tabular}{|l |l |b | b | b | b|b|b|}
\hline
&  &  &KECM  &KECM  & &Pairwise &Pairwise  \\
$\zeta $      & $\sigma_j^{2}$  & KEM     & Laplace & Spike     & MCMC &Refresh &Refresh \\
& & (sec) & (sec)& Slab(sec) &(sec) &(sec) & jump (sec) \\
\hline
  1 & N/A & 24.9 &   77.6 &  57.5 & 182.5 &0.7 &3.1\\ \hline
 0.9999 & 6.25e-06&  28.7  & 76.4 &  58.0 & 182.4 &0.7 &3.1 \\ \hline
 0.9999 & 0.0001 &  48.9 &  83.5 &  59.8 & 184.0 &0.7 &3.1 \\ \hline 
 0.9995 & 6.25e-06 &  46.2 &  88.5 &  61.5 & 185.1 &0.7 &3.2\\ \hline
  0.9995 & 0.0001 &  95.1 & 109.9 &  64.4 & 193.8 &0.7 &3.2 \\ \hline
 0.999 & 6.25e-06 &  51.9 &  83.1 &  62.0 & 187.4 &0.8 &3.2\\ \hline
 0.999 & 2.5e-05 &  86.9 &  99.9 &  66.1 & 177.5 &0.8 &3.3\\ \hline
 0.999 & 0.0001 &  90.0 & 122.9 &  72.8 & 200.9 &0.8 & 3.2\\ \hline 
\end{tabular}
\caption{Run-time (seconds) for GARCH-jump model with stochastic microstructure noise variance}
\label{TableTiming}
\end{table}

\subsection{Numerical Results Summary}
The following are key observations from the numerical simulation results:
\begin{enumerate}
  \item Both KECM and MCMC approaches outperform KEM in the presence of jumps.
  \item Laplace prior underperforms spike and slab models for large jumps.
  \item Spike and slab models are more robust to stochastic microstructure noise variance than the Laplace prior model.
\end{enumerate}

The first observation is not surprising since both the KECM and MCMC approaches explicitly account for jumps.   The second and third observations may be the result of a large jump estimation bias that can occur when using the Laplace prior for large $\sigma_{j}^{2}$.

\section{Conclusion}
This work has introduced two jump robust KECM methods for estimating asset return covariance from high-frequency data.   The methods address several stylized facts found in high frequency data: 1) asynchronous returns, 2) market microstructure noise, and 3) jumps.   The first method, a KECM approach, was derived using both Laplace and spike and slab distributed jump models.  The second method utilized a MCMC approach to approximate the posterior mean of the covariance estimate.   Here the jumps were modeled using a spike and slab distribution.

Both proposed techniques improve covariance estimation performance versus existing methods when jumps are present and are robust to other stylized facts such as volatility clustering and non-stationary microstructure noise.   When comparing the spike and slab and Laplace jump models, the spike and slab approach demonstrated more robustness especially to larger jumps and stochastic microstructure noise variance.   The MCMC approach also shows a modest performance improvement versus the KECM methods when frequent large jumps occur.

As future work other jump models besides spike and slab and Laplace can be considered.  For example both the spike and slab and Laplace priors create a bias in the jump estimates.  The use of other penalties which induce less penalty for large jumps may reduce this bias and improve estimation performance.   Another area that can be addressed is global convergence of the KECM algorithms.   Since the KECM does not necessarily converge to a globally optimal solution additional performance gains may be achievable by attempting multiple initializations or other approaches.

\appendix
    \begin{center}
      {\bf APPENDIX}
    \end{center}

\section{Kalman Smoothing Equations}\label{secAppendixSmooth}
The Kalman smoother can be used to compute the posterior distribution of $X(t)$ given $Y$ and an estimate of $\Theta= [D,\Gamma,\Sigma_{o}',J]$.
From \cite{Shumway} the posterior distribution is normal and is completely characterized by the following quantities for $m=T$
\begin{equation*}
  \bar{X}(t|m) = \E(X(t)|y(1:m))
\end{equation*}
\begin{equation*}
  P(t|m) = \cov(X(t),X(t)|y(1:m))
\end{equation*}
\begin{equation*}
  P(t,t-1|m) = \cov(X(t),X(t-1)|y(1:m)).
\end{equation*}
These values can be computed efficiently using a set of well known forward and backward recursions \cite{ShumwayBook} known as the Rauch-Tung-Striebel (RTS) smoother.   The forward recursions
are
\begin{eqnarray}\label{eq:ForwardEq}
  \bar{X}(t|t-1) &=&\bar{X}(t-1|t-1)+D+J(t) \\
  P(t|t-1) &=& P(t-1|t-1)+\Gamma \\
  G(t)&=&P(t|t-1)I(t)^{T}\left(I(t)P(t|t-1)I(t)^{T} + \Sigma_{o}^{2}(t)\right)^{-1} \\
  \bar{X}(t|t) &=&\bar{X}(t|t-1) +G(t)(y(t)-I(t)\bar{X}(t|t-1)) \\
  P(t|t)&=&P(t|t-1) - G(t)I(t)P(t|t-1)
\end{eqnarray}
with $\bar{X}(0|0)=\mu$ and $P(0|0)=K$.

The backward equations are given by
\begin{eqnarray*}
  H(t-1)&=&P(t-1|t-1)P(t|t-1)^{-1} \\
  \bar{X}(t-1|T)&=&\bar{X}(t-1|t-1) + H(t-1)(\bar{X}(t|T)-\bar{X}(t|t-1)) \\
  P(t-1|T)&=&P(t-1|t-1) \\
  &&+ H(t-1)(P(t|T)-P(t|t-1))H(t-1)^{T}.
\end{eqnarray*}
A backward recursion for computing $P(t,t-1|T)$ is
\begin{eqnarray*}
  P(t-1,t-2|T)&=& P(t-1|t-1)H(t-2)^{T} \\
  & &\; + H(t-1)\left(P(t,t-1|T)-P(t-1|t-1) \right)H(t-2)^{T} \nonumber
\end{eqnarray*}
where
\begin{equation}\label{eq:BackwardRecursion}
  P(T,T-1|T)=(I -G(T)I(T))P(T-1|T-1).
\end{equation}

\section{Derivation of Equation \eqref{eqPosteriorECM}}\label{secAppendixDerive}
Here we derive the expression for
\begin{equation*}
\mathcal{G}(\theta,\Theta^{(j)}) = \mathbb{E}_{p(x|y,\Theta^{(j)})}\log p(X(1:T),y(1:T)|\theta)+\log(p(\theta))
\end{equation*}
given in equation \eqref{eqPosteriorECM}.   First recall the equation for the log-likelihood
\begin{eqnarray}\label{Eq44_1}
  \log p(x,y|\tilde{\theta}) &=& -0.5\sum_{t=1}^{T}\log(|\Sigma_{o}(t)|)  -\frac{1}{2}\sum_{t=1}^{T}||y(t)-\tilde{I}(t)\bar{X}(t)||_{\textnormal{diag}(\Sigma_{o}(t)^{-1}),\ell_{2}}^{2} \nonumber \\
   && - \frac{T-1}{2}\log(|\Gamma|) \nonumber \\
   && - \frac{1}{2}\sum_{t=2}^{T}r(t)^{T}\Gamma^{-1}r(t) \nonumber \\
   && + const
\end{eqnarray}
where
\begin{equation*}
   r(t)=x(t)-x(t-1)-d-j(t).
\end{equation*}
First note that using the relation
\begin{equation*}
  Y(t)-\tilde{I}(t)X(t)  = Y(t)-\tilde{I}(t)(X(t)-\bar{X}(t))-\tilde{I}(t)\bar{X}(t)
\end{equation*}
we have that
\begin{eqnarray}\label{Eq44Other}
\mathbb{E}_{p(x|y,\Theta^{(j)})} ||(y(t)-\tilde{I}(t)X(t))||_{\textnormal{diag}(\Sigma_{o}(t)^{-1}),\ell_{2}}^{2} &=& ||y(t)-\tilde{I}(t)\bar{X}(t)||_{\textnormal{diag}(\Sigma_{o}(t)^{-1}),\ell_{2}}^{2}+ \nonumber \\ &&+\trace(P(t|T)\tilde{I}(t)^{T}\Sigma_{o}(t)^{-1}\tilde{I}(t)).
\end{eqnarray}
Similarly noting that for $R(t) \doteq X(t)-X(t-1)-d-j(t)$
\begin{equation}
   R(t)=X(t)-\bar{X}(t) - (X(t-1) -\bar{X}(t-1)) +(\bar{X}(t)-\bar{X}(t-1)) -d-j(t) \\
\end{equation}
we can show that
\begin{eqnarray} \label{Eq44Sum}
\sum_{t=2}^{T}\mathbb{E}_{p(x|y,\Theta^{(j)})} R(t)^{T}\Gamma^{-1}R(t) =\textnormal{tr}(\Gamma^{-1}(C-B-B^{T}+A)).
\end{eqnarray}
using the orthogonality principle.  From equations \eqref{Eq44Other} and \eqref{Eq44Sum} we arrive at \eqref{eqPosteriorECM}.

\section{Convergence of KECM Algorithms}\label{Sec:ConvergeECM}
Convergence of the EM and ECM algorithms in general is considered in \cite{WuEM} and \cite{ECMRubin} respectively.  It is shown in \cite{ECMRubin} that the ECM algorithm converges to stationary point of the log posterior under the following mild regularity conditions
\begin{enumerate}
  \item \label{Cond:ECMA} Any sequence $\Theta^{(k)}$ obtained using the ECM algorithm lies in a compact subset of the parameter space, $\Omega$.   For our case we need to restrict the parameter space such that $\sigma_{o}^{2} \ne 0$ and $\Gamma$ is positive definite.
  \item \label{Cond:ECMB} $\mathcal{G}(\Theta,\Theta')$ is continuous in both $\Theta$ and $\Theta'$.
  \item \label{Cond:ECMC} The log posterior $L(\Theta)$ is continuous in $\Omega$ and differentiable in the interior of $\Omega$.
\end{enumerate}
\subsection{Algorithm \ref{Alg:ECM}}
Since the Laplace prior on $J$ is not differentiable condition \ref{Cond:ECMC} is not satisfied and the results in \cite{ECMRubin} are not directly applicable.   However the proofs and solution set in \cite{ECMRubin} can be modified to handle this irregularity.

Before addressing condition \ref{Cond:ECMC} we first verify condition \ref{Cond:ECMA}.  We start by examining the sequence of covariance estimates $\Gamma^{(k)}$.
\begin{lemma}\label{Lem:Gamma}
Assume a noisy asset price is observed at least one time for each asset for $t>1$ and that $\tilde{I}(t) \ne 0$ for all $t$.   Let $\Gamma^{(k)}$ be a sequence of solutions obtained with Algorithm \ref{Alg:ECM}, where $\Gamma^{(0)}$ is positive definite.  Then sequences $\Gamma^{(k)}$ and $\frac{1}{s^{(k)}}$ are bounded where $s^{(k)}$ is the minimum eigenvalue of $\Gamma^{(k)}$.  In addition the sequence $\sigma_{o,i}^{2,(k)}$ is bounded below and above by positive values for all $i$.
\end{lemma}
\begin{proof}

Since $W_{o}$ is positive definite we have from equation \eqref{GammaEM} that $s^{(k)}$ is bounded below by a positive constant which implies $\frac{1}{s^{(k)}}$ is bounded.  Similarly by equation \eqref{obsEM} we have $\sigma_{o,i}^{2,(k)}$ is bounded below by a positive constant.
To prove that $\Gamma^{(k)}$ is bounded we note that the posterior may be written as
\begin{eqnarray}
  p(\theta|y)&=&C_{1}p(y|\theta)p(\theta) \nonumber \\
&=&C_{1}p(y(1)|\theta)p(\theta)\prod_{t=2}^{T}p(y(t)|y(1:t-1),\theta) \nonumber \\
&\le& C_{2} p(y(1)|\theta)\prod_{t=2}^{T}p(y(t)|y(1:t-1),\theta) \nonumber \label{ineqECMconverge}
\end{eqnarray}
where $C_{1}$ is a constant not dependent on $\theta$ and where $C_{2} = C_{1}\sup_{\theta}p(\theta)$.  Note that $C_{2} < \infty$.

For $t>1$ each of the conditional distributions  $p(y(t)|y(1:t-1),\theta)$ is a normal distribution with covariance
\begin{equation*}
  Q(t)=\tilde{I}(t)P(t|t-1)\tilde{I}(t)^{T}+\sigma_{o}^{2}I
\end{equation*}
where for notational simplicity we suppress the dependence of $Q(t)$ and $P(t|t-1)$ on $k$.   Since $\sigma_{o,i}^{2}$ is bounded below by a positive value, it follows that $\frac{1}{|Q(t)|}$ is bounded.

Now suppose that $\Gamma^{(k)}$ is unbounded.   Then since
 \begin{equation*}
   P(t|t-1) = P(t-1|t-1)+\Gamma
 \end{equation*}
  $P(t|t-1)$ is unbounded as $k$ goes to $\infty$.   Since an observation of each asset's price occurs at least once for $t>1$ it follows that $Q(\tau)$ is unbounded (as $k \rightarrow \infty$) for some $\tau>1$.   Then since the smallest eigenvalue of $Q(\tau)$ is bounded below by a positive constant, the determinant of $Q(\tau)$ is unbounded.  Thus a subsequence of $p(y(\tau)|y(1:\tau-1),\Theta^{(k)})$ will approach 0.  Since $\frac{1}{|Q(t)|}$ is bounded, $p(y(t)|y(1:t-1),\Theta^{(k)})$ will remain bounded above for all $t$.   Then using \eqref{ineqECMconverge} we have
\begin{eqnarray*}
  p(\theta|y) &\le& C_{2} p(y(1)|\theta)\prod_{t=2}^{T}p(y(t)|y(1:t-1),\theta)  \\
  &=& C_{2} p(y(\tau)|y(1:\tau-1),\theta)\prod_{t\ne \tau}^{T}p(y(t)|y(1:t-1),\theta)
\end{eqnarray*}
which implies a subsequence of $p(\Theta^{(k)}|y)$ will converge to 0.   This contradicts the monotonicity of the ECM algorithm \cite{ECMRubin}.
The proof that the sequence $\sigma_{o,i}^{2,(k)}$ is bounded above for all $i$ is similar.
\end{proof}
\begin{lemma}
Assume the conditions of Lemma \ref{Lem:Gamma}.   Let $\lambda(t)^{-1,(k)}$ be a sequence of solutions obtained with Algorithm \ref{Alg:ECM} where $\Gamma^{(0)}$ is positive definite.   Then there exist finite positive numbers $a,b$ where $a\le \lambda_{i}(t)^{(k)} \le b$ for all $t,k$ and $i$.
\end{lemma}
\begin{proof}
By the update equation \eqref{eqLamUp} we may set $b=\frac{\alpha_{\lambda}+2}{\beta_{\lambda}}$ which is positive and finite.  By way of contradiction assume the lower bound does not hold.  Then for some $i$ and $t$ there exists a subsequence $\lambda_{i}(t)^{(k_{n})}$ such that $\lim_{n \rightarrow \infty} \lambda_{i}(t)^{-1,(k_{n})} =\infty$.   Since each $\lambda_{i}(t)^{-1}$ is the mode of an inverse gamma distribution it follows that the posterior scale parameter,($\beta_{\lambda}+|j^{(k_{n})}|$) goes to infinity . This implies that $p(\lambda_{i}(t)^{-1,(k_{n})},j_{i}(t)^{(k_{n})}) \rightarrow 0$.   Since each prior density function is bounded as $\lambda_{i}(t) \rightarrow 0$ this implies that $p(\theta)$ goes to zero, contradicting the monotonicity of the ECM algorithm.   Thus there exists an $a>0$ such that $\lambda_{i}(t)^{(k)}>a$ for all $t,k$ and $i$.
\end{proof}

Now we prove that the sequences $J^{(k)}$ and $D^{(k)}$ are also well behaved.
\begin{lemma}
Assume the conditions of Lemma \ref{Lem:Gamma}.   Let $J^{(k)}$ and $D^{(k)}$ be sequences of solutions obtained with Algorithm \ref{Alg:ECM} where $\Gamma^{(0)}$ is positive definite.  Then sequences $J^{(k)}$ and $D^{(k)}$ are bounded.
\end{lemma}
\begin{proof}
From Lemma \ref{Lem:Gamma} the likelihood $p(y|\theta)$ is bounded above.  Recall from the previous lemma that there exists an $a>0$ such that for all $k$, $\lambda_{i}(t)^{(k)} \ge a$.  Since the prior density function is bounded above for each parameter it follows that $\lim_{j \rightarrow \infty} p(\theta) =0$.   This implies $J^{(k)}$ is bounded by the monotonicity of the ECM algorithm.   Since $\lim_{d \rightarrow \infty} p(\theta) =0$ it also follows that $D^{(k)}$ is bounded.
\end{proof}

The above lemmas imply the following corollary.
\begin{corr}\label{corr:ECM1}
The sequence $\Theta^{(k)}$ is bounded and all limit points are feasible ( e.g. variance non-zero, positive definite covariance).
\end{corr}

Now we derive some additional properties of the limit points of $\Theta^{(k)}$.   To do this we shall refer to Zangwill's convergence theorem \cite{ZangwillBook}.   To use Zangwill's theorem, we first define $\mathcal{A}$ to be a point to set mapping  defined by the ECM algorithm i.e. $\Theta^{(k+1)} \in \mathcal{A}(\Theta^{(k)})$.   Let us define a solution set, $\mathcal{S}$, as the set of $\theta$ such that
\begin{eqnarray*}
  \theta_{1} = \arg\max_{v} \mathcal{G}\left(\left[v,\theta_{2},\theta_{3},\theta_{4},\theta_{5} \right],\theta \right) \\
  \theta_{2} = \arg\max_{v} \mathcal{G}\left(\left[\theta_{1},v,\theta_{3},\theta_{4},\theta_{5} \right],\theta \right) \\
  \theta_{3} = \arg\max_{v} \mathcal{G}\left(\left[\theta_{1},\theta_{2},v,\theta_{4},\theta_{5} \right],\theta \right) \\
  \theta_{4} = \arg\max_{v} \mathcal{G}\left(\left[\theta_{1},\theta_{2},\theta_{3},v,\theta_{5} \right],\theta \right) \\
  \theta_{5} = \arg\max_{v} \mathcal{G}\left(\left[\theta_{1},\theta_{2},\theta_{3},\theta_{4},v \right],\theta \right). \\
\end{eqnarray*}
By definition $\theta \in \mathcal{A}(\theta)$ for all $\theta \in \mathcal{S}$.  This along with the the monotonicity of the ECM algorithm implies that $L(\theta)$ is an ascent function, i.e.
\begin{eqnarray*}
  L(\theta')> L(\theta)  \textnormal{ for all } \theta \not\in \mathcal{S}, \theta' \in \mathcal{A}(\theta) \nonumber \\
  L(\theta')\ge L(\theta)  \textnormal{ for all } \theta \in \mathcal{S}, \theta' \in \mathcal{A}(\theta).\nonumber
\end{eqnarray*}
Since $\mathcal{G}(\theta,\theta')$ is continuous in both $\theta$ and $\theta'$ we have that $\mathcal{A}$ is a closed mapping.  Thus we have the following theorem.
\begin{theorem}\label{thm:ECMconverge1}
All limit points of $\Theta^{(k)}$ belong to $\mathcal{S}$.
\end{theorem}
\begin{proof}
This is a direct consequence of Zangwill's convergence theorem \cite{ZangwillBook} (also known as the Global convergence theorem \cite{LuenbergerBook}).   To invoke the theorem we must meet the following conditions
\begin{itemize}
   \item $\Theta^{(k)}$ belongs to a compact subset of the feasible solutions
   \item $\mathcal{A}$ is closed
   \item There exists a continuous ascent function
 \end{itemize}
 All three of these conditions were shown above, thus the theorem follows from Zangwill's convergence theorem.
\end{proof}

Now we show that if $\theta' \in \mathcal{S}$ then $\theta'$ is in some sense a ``stationary'' point of the log posterior $L(\theta)=\log p(\theta|y)$.
\begin{theorem}\label{Thm:ECMconvLaplace}
Let $\theta' \in \mathcal{S}$.   Then
\begin{equation*}
  \nabla_{\theta_{i}} L(\theta)_{|\theta=\theta'} =0 \textnormal{ for } i \in {1,2,3,5}
\end{equation*}
and
\begin{equation*}
  0 \in \partial_{\theta_{4}} L(\theta)_{|\theta=\theta'}.
\end{equation*}
\end{theorem}
\begin{proof}
To show this we first note that $L(\theta)$ can be written as \cite{ECMRubin}
\begin{equation*}
  L(\theta|y) = \mathcal{G}(\theta,\theta') - H(\theta,\theta')
\end{equation*}
where
\begin{equation*}
  H(\theta,\theta') = \E_{p(x|y,\theta')}\log p(X|y,\theta).
\end{equation*}
From the information inequality we have that $H(\theta',\theta') \ge H(\theta,\theta')$ for all feasible $\theta$.  Since $H(\theta,\theta')$ is differentiable with respect to $\theta$ it follows that
\begin{equation*}
  \nabla_{\theta} H(\theta,\theta')_{|\theta=\theta'} = 0.
\end{equation*}
Since $\nabla_{\theta_{i}}\mathcal{G}(\theta,\theta')_{|\theta=\theta'}=0$ \textnormal{ for } $i \in {1,2,3,5}$
it follows that
\begin{equation*}
  \nabla_{\theta_{i}} L(\theta)_{|\theta=\theta'} =0 \textnormal{ for } i \in {1,2,3}
\end{equation*}
Also since $\mathcal{G}(\theta,\theta')$ and $H(\theta,\theta')$ are convex in $j$, and $\theta' \in \mathcal{S}$, it follows that
\begin{equation*}
  0 \in \partial_{\theta_{4}} \mathcal{G}(\theta,\theta')
\end{equation*}
 which implies
 \begin{equation*}
  0 \in \partial_{\theta_{4}} L(\theta,\theta').
\end{equation*}
\end{proof}
\subsection{Algorithm \ref{Alg:ECMspike}}
Analogous results to Corollary \ref{corr:ECM1} and Theorem \ref{thm:ECMconverge1} may proven for Algorithm \ref{Alg:ECMspike} using same arguments as Algorithm \ref{Alg:ECM}.  The following result is analogous to Theorem \ref{Thm:ECMconvLaplace}.
\begin{theorem}
Let $\theta' \in \mathcal{S}$ where $\mathcal{S}$ is the set of fixed points of the Algorithm \ref{Alg:ECMspike}.    Then
\begin{equation*}
  \nabla_{\theta_{i}} L(\theta)_{|\theta=\theta'} =0 \textnormal{ for } i \in {1,2,3,5,6}.
\end{equation*}
\end{theorem}
The proof of this result is the same as Theorem \ref{Thm:ECMconvLaplace}. 
\section{MCMC Details}\label{SecMCMCAppendix}
In this section we state the conditional distributions needed to implement the Gibbs sampling approach in Section \ref{SecGibbs}.
\subsection{Conditional Price Distribution}
Let $\mathcal{N}(x,\mu,R)$ be the normal PDF in $x$ with mean $\mu$ and covariance $R$.  For the Gibbs sampling approach we need to determine the conditional distribution of $X(t)$ given $\Phi_{-1},\Gamma,Y$, and $X(s) s \ne t$.   Let $Y_{tot}(t)$ be the total price vector obtained from observed prices $Y(t)$ and the current sample of the unobserved prices $Y_{miss}(t)$.   We first note that for $t >1, t<T$
\begin{eqnarray*}
  p(x(t)|x(s),\phi_{-1},\gamma,y_{tot};\forall s\ne t)&=&p(x(t)|x(t-1),x(t+1),\phi_{-1,-0},\gamma,y_{tot}(t))\\
  &\propto& p(x(t+1)|x(t),\phi_{-1,-0},\gamma)\\
  &&p(y_{tot}(t)|x(t),\phi_{-1,-0},\gamma)\\
  &&p(x(t)|x(t-1),\phi_{-1,-0},\gamma).
\end{eqnarray*}
By properties of normal distributions
\begin{equation}
  p(x(t+1)|x(t),\phi_{-1,-0},\gamma)=\mathcal{N}(x(t+1),x(t)+j(t+1)+d,\Gamma)
\end{equation}
and
\begin{equation}
  p(y_{tot}(t)|x(t),\phi_{-1},\gamma)=\mathcal{N}(y_{tot}(t),x(t),\sigma_{o}^{2}I).
\end{equation}

With this recall the following multiplication property of normal PDFs
\begin{equation*}
  \mathcal{N}(x,\mu_{1},R_{1})\mathcal{N}(x,\mu_{2},R_{2}) \propto \mathcal{N}(x,\mu_{3},R_{3})
\end{equation*}
where
\begin{equation*}
  R_{3}=(R_{1}^{-1}+R_{2}^{-1})^{-1}
\end{equation*}
and
\begin{equation*}
  \mu_{3}=R_{3}R_{1}^{-1}\mu_{1} + R_{3}R_{2}^{-1}\mu_{2}.
\end{equation*}
Using the multiplication property above
\begin{equation*}
  p(x(t)|x(t-1),y_{tot}(t))=\mathcal{N}(x(t),q,Q)
\end{equation*}
where
\begin{equation*}
  Q=(\Gamma^{-1}+\sigma_{j}^{-2}I)^{-1}
\end{equation*}
\begin{equation*}
  q=Q\Gamma^{-1}(x(t-1)+J(t)+D) + \sigma_{j}^{-2}Qy_{tot}(t).
\end{equation*}
Applying the multiplication property again gives
\begin{equation*}
  p(x(t)|x(t-1),y_{tot}(t),x(t+1))=\mathcal{N}(x(t),q',Q')
\end{equation*}
where
\begin{equation*}
  Q'=(\Gamma^{-1}+Q^{-1})^{-1}
\end{equation*}
and
\begin{equation*}
  q'=Q'Q^{-1}q + Q'\Gamma^{-1}(x(t+1)-D-J(t+1).
\end{equation*}
The conditional distributions for $t=1$ and $t=T$ can be derived similarly.

Another  approach to sampling from the conditional distribution $p(X|Y,\Phi)$ is the Forward Filtering Backward Simulation (FFBS) approach \cite{FFBS}.   The FFBS algorithm allows one to sample directly from the conditional joint distribution of $X(1:T)$, but the required backward simulation can be computationally intensive as one must compute $T$ Cholesky decompositions.   In the approach outlined above one only needs to compute 3 Cholesky decompositions ($t=1$,$t=T$ and once for $1<t<T$).

\subsection{Conditional Jump Distribution}\label{SecMCMCAppendixJump}
Let $\mathcal{N}(x,\mu,\tau)$ be the normal PDF in $x$ with mean $\mu$ and variance $\tau$.
Recall that the prior distribution of the jumps is the spike and slab prior
\begin{equation*}
  p(j)= f(j)= \zeta\delta_{0}(j) + (1-\zeta)\mathcal{N}(j,0,\sigma_{j}^{2})
\end{equation*}
and that in the prior distribution the jumps are independent and identically distributed.   When conditioned on $\sigma_{j}^{2},\zeta,X, D$ and $\Gamma$, $J_{m}(t)$ and $J_{n}(s)$ remain independent for $s \ne t$ but  $J_{m}(t)$ and $J_{n}(t)$ become dependent.   The conditional distribution $p(j(t)|\phi_{-2},\gamma)$ can be written as
\begin{equation}\label{appenJumpProb}
  p(j(t)|\phi_{-2},\gamma)= \frac{c \exp\left(\frac{-(j(t)-v(t))^{T}\Gamma^{-1}(j(t)-v(t))}{2}\right)}{\sqrt{(2\pi)^{N}|\Gamma|}} \prod_{i=1}^{N}f(j_{i}(t))
\end{equation}
where $v(t)=X(t)-X(t-1)-D$ and where $c>0$ is independent of $j(t)$.
Sampling directly from this distribution is difficult due to the combinatorial nature of the prior.   Therefore we sample sequentially each component of $j_{i}(t)$ conditioned on $j_{-i}(t)$.

To derive the posterior distribution of $j_{i}(t)$ we note that from properties of the multivariate normal distribution
\begin{equation*}
  p(j_{i}(t)|j_{-i}(t),\gamma,\phi_{-2}) \propto \mathcal{N}(j_{i}(t),a(i),b^{2}(i))f(j_{i}(t))
\end{equation*}
where
\begin{equation*}
  a(i)=v_{i}(t)+\Gamma_{i,-i}\Gamma_{-i,-i}^{-1}(j_{-i}(t)-v_{-i}(t))
\end{equation*}
and
\begin{equation*}
  b^{2}(i)=\Gamma_{i,i}-\Gamma_{i,-i}\Gamma_{-i,-i}^{-1}\Gamma_{-i,i}.
\end{equation*}
Next we determine $Pr(Z_{i}(t)=z|\phi_{-2},\gamma,j_{-i}(t))$ for $z=0,1$.  Recall the following identity for normal PDFs
\begin{equation*}
  \mathcal{N}(x,u_{1},\tau_{1}^{2})\mathcal{N}(x,u_{2},\tau_{1}^{2})= \mathcal{N}(u_{1},u_{2},\tau_{1}^{2}+\tau_{2}^{2})\mathcal{N}(x,u,\tau^{2})
\end{equation*}
\begin{equation*}
  u=\frac{\tau_{1}^{-2}u_{1}+\tau_{2}^{-2}u_{2}}{\tau_{1}^{-2}+\tau_{2}^{-2}}
\end{equation*}
\begin{equation*}
  \tau^{2}=\frac{\tau_{1}^{2}\tau_{2}^{2}}{\tau_{1}^{2}+\tau_{2}^{2}}.
\end{equation*}
Using the relationship above we have
\begin{equation}\label{eq:ZMCMCpost}
  Pr(Z_{i}(t)=z|\phi_{-2},\gamma,j_{-i}(t)) \propto \begin{cases} \zeta \mathcal{N}(0,a(i),b^{2}(i))\mbox{ if } z = 0\\
  (1-\zeta)\mathcal{N}(0,a(i),b(i)+\sigma_{j}^{2}) \mbox{ if } z = 1
 \end{cases}
\end{equation}
We now draw $Z_{i}(t)$ from this distribution.   If $Z_{i}(t)=0$, $J_{i}(t)$ is set to zero, otherwise we draw $J_{i}(t)$ from the distribution
\begin{equation} \label{eq:jMCMCpost}
 p(j_{i}(t)|z_{i}(t)=1,\phi_{-2},\gamma,j_{-i}(t))
\end{equation}
which from the above relationship is a normal distribution with mean
\begin{equation*}
  \frac{a(i)}{1+b^{2}(i)\sigma_{j}^{-2}}
\end{equation*}
and variance
\begin{equation*}
  \frac{b^{2}(i)\sigma_{j}^{2}}{b^{2}(i)+\sigma_{j}^{2}}.
\end{equation*}

\subsubsection{Conditional Posterior Mode of $j_{i}$ in KECM spike and slab model}
Note that the conditional maximization steps for $J$ used in KECM algorithm for spike and slab models can be derived in a similar manner as above.  To see this note that \eqref{eqJumpProbSpike} is up to a constant the logarithm of \eqref{appenJumpProb} where $v$ is replaced with $\Delta$.   Thus one can compute the modes of $Z_{i}$ and $J_{i}$ from the conditional distributions defined above in \eqref{eq:ZMCMCpost} and \eqref{eq:jMCMCpost}.
\subsection{Other conditional distributions}
The remaining conditional distributions for the other parameters are easily obtained due to conjugate prior relationships \cite{FinkConjPrior}.



\bibliography{C:/Users/Micha/Documents/jXin/papers/trunk/siamltexmm/proposal}
\bibliographystyle{siam}

\end{document}